\documentclass[final,twocolumn,5p,authoryear]{elsarticle}
\makeatletter
\def\ps@pprintTitle{%
   \let\@oddhead\@empty
   \let\@evenhead\@empty
   \gdef\@oddfoot{\footnotesize This article is published in {\it Automatica}, 105: 216--227, 2019. Minor corrections and modifications appear in blue colored text.\hfill}
   \let\@evenfoot\@oddfoot
}
\makeatother

\usepackage{amsmath,amsthm,amssymb}
\usepackage{todonotes}
\usepackage{yhmath}
\usepackage{ulem}
\usepackage{enumitem}
\usepackage{hyperref}
\usepackage{tikz}
\usetikzlibrary{circuits.ee.IEC}
\usetikzlibrary{shapes,arrows}
\newlist{lyap}{enumerate}{1}
\setlist[lyap]{label = {\bf(L\arabic*)}, resume}

\newlist{assum}{enumerate}{1}
\setlist[assum]{label = {\bf(A\arabic*)}, resume}

\newlist{design}{enumerate}{1}
\setlist[design]{label = {\bf(D\arabic*)}, resume}

\theoremstyle{plain}
\newtheorem{thm}{Theorem}
\newtheorem{cor}{Corollary}
\newtheorem{prop}{Proposition}

\theoremstyle{definition}
\newtheorem{defn}{Definition}

\theoremstyle{lemma}

\theoremstyle{remark}
\newtheorem{remark}{Remark}

\DeclareMathOperator{\dom}{dom}
\DeclareMathOperator*{\esssup}{ess-sup}

\newcommand{\R}{\mathbb{R}}
\newcommand{\N}{\mathbb{N}}
\newcommand{\Z}{\mathbb{Z}}
\newcommand{\Rposi}{\R_{\geq 0}}

\newcommand{\cF}{\mathcal{F}}
\newcommand{\cG}{\mathcal{G}}

\newcommand{\cK}{\mathcal{K}}
\newcommand{\cL}{\mathcal{L}}
\newcommand{\cM}{\mathcal{M}}

\newcommand{\cKinfty}{\mathcal{K}_\infty}
\newcommand{\cP}{\mathcal{P}}
\newcommand{\cA}{\mathcal{A}}

\newcommand{\cC}{\mathcal{C}}
\newcommand{\cS}{\mathcal{S}}
\newcommand{\cD}{\mathcal{D}}
\newcommand{\cX}{\mathcal{X}}
\newcommand{\ul}{\underline}
\newcommand{\ol}{\overline}
\newcommand{\demi}{\frac{1}{2}}
\newcommand{\qtr}{\frac{1}{4}}
\newcommand{\acplowbar}{\underline{a}_{c,p}}
\newcommand{\acpupbar}{\overline{a}_{c,p}}
\newcommand{\aoplowbar}{\underline{a}_{o,p}}
\newcommand{\aopupbar}{\overline{a}_{o,p}}
\newcommand{\gmacpupbar}{\overline{\gamma}_{c,p}}
\newcommand{\gmaopupbar}{\overline{\gamma}_{o,p}}
\newcommand{\gmaupbar}{\overline{\gamma}_p}
\newcommand{\rhoopupbar}{\overline{\rho}_{o,p}}
\newcommand{\rhocpupbar}{\overline{\rho}_{c,p}}
\newcommand{\sigmaopupbar}{\overline{\mu}_{o,p}}
\newcommand{\sigmacpupbar}{\overline{\mu}_{c,p}}

\newcommand{\eps}{\varepsilon}
\newcommand{\et}{\mathtt{event}}
\newcommand{\true}{\mathtt{true}}
\newcommand{\sw}{\mathtt{sw}}

\newcommand{\sigetao}{\mu_{o,p}(\eta_o)}
\newcommand{\sigetac}{\mu_{c,p}(\eta_c)}

\newcommand{\inv}{^{\raisebox{.2ex}{$\scriptscriptstyle-1$}}}
\newcommand\norm[1]{\left\lVert#1\right\rVert}

\begin{document}
%%%% For figure
\tikzstyle{block} = [draw, rectangle, 
    minimum height=3em, minimum width=6em]
\tikzstyle{sum} = [draw, circle, node distance=1cm]
\tikzstyle{input} = [coordinate]
\tikzstyle{output} = [coordinate]
\tikzstyle{pinstyle} = [pin edge={to-,thin,black}]

\begin{frontmatter}

\title{ISS Lyapunov Functions for Cascade Switched Systems and Sampled-Data Control}%\tnoteref{mytitlenote}}

\author[gxzaddress]{GuangXue Zhang}
\ead{gxzhang7@uci.edu}

\author[ataddress]{Aneel Tanwani\corref{mycorrespondingauthor}}
\ead{aneel.tanwani@laas.fr}
\address[gxzaddress]{Department of Aerospace Engineering, University of California at Irvine, USA}
\address[ataddress]{Laboratory for Analysis and Architecture of Systems, Toulouse, CNRS, France}
\cortext[mycorrespondingauthor]{Corresponding author}
\fntext[myfootnote1]{The work of G. Zhang was supported under ``IDEX grant for nouveaux entrants'' as a part of her Masters internship at LAAS.
The work of A. Tanwani is partially supported by ANR JCJC project ConVan with grant number ANR-17-CE40-0019-01.
}
%\fntext[myfootnote2]{Researcher in LAAS-CNRS, France}


%\author[ataddress]{Global Customer Service\corref{mycorrespondingauthor}}

\begin{abstract}
Input-to-state stability (ISS) of switched systems is studied where the individual subsystems are connected in a serial cascade configuration, and the states are allowed to reset at switching times. An ISS Lyapunov function is associated to each of the two blocks connected in cascade, and these functions are used as building blocks for constructing ISS Lyapunov function for the interconnected system. The derivative of individual Lyapunov functions may be bounded by nonlinear decay functions, and the growth in the value of Lyapunov function at switching times may also be a nonlinear function of the value of other Lyapunov functions. The stability of overall hybrid system is analyzed by constructing a newly constructed ISS-Lyapunov function and deriving lower bounds on the average dwell-time. The particular case of linear subsystems and quadratic Lyapunov functions is also studied. The tools are also used for studying the observer-based feedback stabilization of a nonlinear switched system with event-based sampling of the output and control inputs. We design dynamic sampling algorithms based on the proposed Lyapunov functions and analyze the stability of the resulting closed-loop system.
\end{abstract}

\begin{keyword}
Switched systems \sep input-to-state stability \sep cascade connection \sep multiple Lyapunov functions \sep average dwell-time \sep output feedback \sep event-based control
%\MSC[2010] 00-01\sep  99-00
\end{keyword}

\end{frontmatter}

%\linenumbers

\section{Introduction}

Switched systems, or in general, hybrid dynamical systems provide a framework for modeling a large class of physical phenomenon and engineering systems which combine discrete and continuous dynamics. Due to their wide utility, such systems have been extensively studied in the control community over the past two decades; see the books by \cite{Libe03} and \cite{GoebSanf12} for comprehensive overview.
%This has led to the study of several fundamental properties of switched systems, such as stability (see \citep{ShorWirt07} and \citep{LinAnts09} for overviews), observability \citep{TanwShim13}, controllability \citep{SunGe02}, invertibility \citep{TanwLibe10}, under different assumptions on continuous and discrete dynamics.

This article addresses a robust stability problem for systems with switching vector fields and jump maps. In our setup, each subsystem has a two-stage serial cascade structure where the output of first block acts as an input to the second block, and the disturbances we consider are an exogenous input to the first block, see Figure~\ref{fig:cascade}. By proposing a novel construction for multiple Lyapunov functions for such configurations, we analyze the stability of the interconnected switched system by deriving lower bounds on average dwell-time between switching times. It is seen that such a configuration arises in the context of output feedback stabilization of switched systems with known switching signal where the inputs and outputs are time-sampled. The theoretical tools developed in the earlier part of this paper are then used to design sampling algorithms and analyze stability of the resulting sampled-data system. A preliminary version of the sampled-data problem, studied in the later part of this paper, has appeared in \citep{ZhanTanw18}.

Stability of switched systems has been a topic of interest in control community for past two decades now. Depending on the class of switching signals, or the assumptions imposed on the continuous dynamics, different approaches have been adopted in the literature to study the convergence of the state trajectories. The book \citep{Libe03} provides an overview on this subject. For our purposes, the approach based on slow switching is more relevant. In this direction, the pioneering contribution comes from \citep{HespMorse99} where the lower bounds on average dwell-time are computed using multiple Lyapunov functions. Another result on slow switching, but with state-dependent average dwell-time, has appeared in \citep{DePeDeSa03}. The second fundamental tool, that we build on, relates to the robustness with respect to external disturbances, formalized by the notion of {\it input-to-state stability} (ISS) introduced in \citep{Sontag89}. Using these classical works as foundation, our article provides a certain construction of the ISS-Lyapunov functions for the switched systems in cascade configuration and develops lower bounds on the dwell-time that guarantee ISS property for the switched system.

One of the first results on input-to-state stability of switched systems appears in \citep{VuChat07}, where the authors associate an ISS-Lyapunov function to each subsystem with linear decay rate, and assume that the Lyapunov function for each subsystem can be linearly dominated by the Lyapunov function of another subsystem at switching times.
Other relevant papers studying ISS property for systems with jump maps using dwell-time conditions are \citep{HespLibe08}, \citep{DashMiron13}. Using the notion of ISS, tools such as small gain theorems \citep{JiangTeel94}, or cascade principles \citep{SontTeel95} are developed to study different applications. The small gain theorems have in particular found utility in the stability analysis of interconnected systems \citep{Ito06}, \citep{DashRuff10}. For hybrid systems, in general, the ISS results using Lyapunov functions appear in \citep{CaiTeel09}, \citep{CaiTeel13}. Their utility is seen in analyzing stability of two interconnected hybrid systems in \citep{Sanf14} and \citep{LibeNesi14}, where the later in particular focuses on small-gain theorems and their application in control over networks. The stability of interconnected switched systems based on small gain theorems also appears in \citep{YangLibe15}. The more recent article then generalizes the results on interconnections \citep{YangLibe16} while allowing for potentially unstable subsystems and jump dynamics.

\begin{figure}
\begin{center}
\begin{tikzpicture}
\draw[thick] (1,0) node [rectangle, rounded corners, draw, minimum height = 0.75cm, text centered] (sys1) {$\begin{aligned}\dot e &= f_{o,\sigma}(e,d) \\ e^+ &= g_o(e,d)\end{aligned}$};
\draw[thick] (4.5,0) node [rectangle, rounded corners, draw, minimum height = 0.75cm, text centered] (sys2) {$\begin{aligned}\dot x &= f_{c,\sigma}(x,e) \\ x^+ &= g_c(x,e)\end{aligned}$};
\draw[thick, ->] (sys1.east) -- node[anchor=south]{$e$} (sys2.west);
\draw[thick, ->] (-1,0) node[anchor=south]{$d$}-- (sys1.west);
\draw[thick, ->] (sys2.east) -- node[anchor=south]{$x$} (6.5,0);
\end{tikzpicture}
\caption{Two-stage serial cascade system with switching dynamics.}
\label{fig:cascade}
\end{center}
\end{figure}
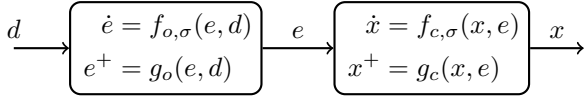

The first part of this article is also built on analyzing the stability of interconnected subsystems with continuous and discrete dynamics. However, we are interested in studying systems where the interconnection is described by a cascade configuration, see Fig.~\ref{fig:cascade}. Using the Lyapunov function construction in \citep{TanwTeel15}, we construct the Lyapunov functions for this cascade interconnection. We then use the framework of hybrid systems to describe the overall system with jump maps, and switching signal with average dwell-time constraints. A novel Lyapunov function is constructed for this hybrid system and the corresponding analysis provides the lower bounds on average dwell-time which yield global asymptotic stability of a certain set. In our approach, we do not require the decay rates of the individual Lyapunov functions to be linear, and the upper bounds on the value of individual Lyapunov function at jump instants may be nonlinear functions of other Lyapunov functions. When studying linear systems as an example, even though we associate quadratic Lyapunov functions to individual subsystems, the Lyapunov function for the overall hybrid system involves a product of the exponential function with a non-quadratic function, which to the best of our knowledge is a novel observation.

We then use these constructions to study the feedback stabilization of switched nonlinear systems when the output measurements and control inputs are time-sampled. Using an observer-based controller, where the estimation error dynamics and the closed-loop system (with static control) are ISS with respect to measurement errors, we rewrite the whole system in cascade configuration where the estimation error drives the state of the controlled plant. The measurement errors are introduced because we only send time-sampled outputs to the controller, and the controller only sends sampled control inputs to the plant. In both cases, the sampled measurements are subjected to a zero-order hold, and thus remain constant until the next sampling instant. Our goal is to derive algorithms to compute sampling algorithms which result in global asymptotic stability of the closed-loop system under the average dwell-time assumptions derived earlier. The event-based sampling strategy that we use is inspired from \citep{TanwTeel15}, where the dynamic filters are introduced. The next sampling instant occurs when the difference between the current value of the output (resp. input) and its last sample is comparatively larger than the value of the dynamic filter's state. Beyond the realm of periodic sampling, stabilization of dynamical systems has been studied subject to various sampling techniques, see for example \citep{HeemJoha12} and \citep{TanwChat18} for recent surveys. Among these methods, event-based control has received attention as an effective means of sampling and various variants of this problem have been studied over the past few years. However, this technique has not yet been studied for switched systems which is the second main contribution of this article.

The remainder of the article is organized as follows: In Section~\ref{sec:prelim}, we provide an overview of basic stability notions and existing results which will be used in this article. The system class of interest is introduced in Section~\ref{sec:cascade}, where we develop the main theoretical results on construction of Lyapunov functions, and developing bounds for average dwell-time. These results are applied in Section~\ref{sec:sampling} to study dynamic feedback stabilization of switched nonlinear systems with sampled-data, and the second main results concerning the design of sampling algorithms and stability analysis of closed-loop system is developed in this section. As an illustration, we provide simulation results for an academic example in Section~\ref{sec:example}, followed by some concluding remarks in Section~\ref{sec:conc}.

\section{Preliminaries}\label{sec:prelim}

In this section, we recall some basic notions of interest which relate to the stability of a hybrid system. For our purposes, it is useful to consider hybrid systems with inputs studied in \citep{CaiTeel13}, which are described  by following inclusions:
\begin{align}\label{eq:genHybSys}
\begin{cases}
\dot \xi \in \cF(\xi, d), \quad (\xi,d) \in \cC, \\
\xi^+ \in \cG(\xi,d),  \quad (\xi,d) \in \cD,
\end{cases} 
\end{align}
where $\xi \in \cX$ is the state and $d \in \R^{\ol d}$ is the external disturbance. The {\it flow set} $\cC \subseteq \cX \times \R^{\ol d}$, and the {\it jump set} $\cD \subseteq \cX \times \R^{\ol d}$ are assumed to be relatively closed in $\cX \times \R^{\ol d}$. The set-valued map $\cF: \cC \rightrightarrows \cX$ describes the continuous dynamics when $\xi$ belongs to the flow set $\cC$. The mapping $\cG:\cD \rightrightarrows \cX$ defines the state reset map, when $\xi$ belongs to the jump set $\cD$.

The solution of the hybrid system~\eqref{eq:genHybSys} is defined on a {\it hybrid time domain}. A set $E \subseteq \R_{\ge 0} \times \Z_{\ge 0}$ is called a compact hybrid time domain if $E = \cup_{j = 0}^J ([t_j,t_{j+1}],j)$ for some finite sequence $0 = t_0 \le t_1 \le \dots \le t_{J+1}$. We say that $E$ is a {\it hybrid time domain} if for each $(T,J)$ in $E$, the set $E \cap [0,T] \times \{0,1,\dots,J\}$ is a compact hybrid time domain. A function defined on a hybrid time domain is called a {\it hybrid signal}. In \eqref{eq:genHybSys}, the disturbance $d$ is a hybrid signal, so that $d(\cdot,j)$ is locally essentially bounded. A {\it hybrid arc} $\xi$ is a hybrid signal for which $\xi(\cdot,j)$ is locally absolutely continuous for each $j \in \Z_{\ge 0}$, and we use the notation $\xi(t,j)$ to denote the value of $\xi$ at time $t$ after $j$ jumps. For a given initial condition $\xi(0,0) \in \cC \cup \cD$, and a hybrid signal $d$, the solution $\xi$ to \eqref{eq:genHybSys} is a hybrid arc if $\dom \xi \subseteq \dom d$, and it holds that i)
for every $j\in \Z_{\ge 0}$ and almost every $t \in \R_{\ge 0}$ such that $(t,j) \in \dom \xi$, we have $(\xi(t, j), d(t, j)) \in \cC$ and $\dot \xi(t, j) \in \cF(\xi(t, j), d(t, j))$;
ii) for $(t, j) \in \dom \xi$ such that $(t,j + 1) \in \dom \xi$, we have $(\xi(t,j),d(t,j)) \in \cD$ and $\xi(t,j + 1) \in \cG(\xi(t, j), d(t, j))$.
It is assumed that the quadruple $(\cF,\cG,\cC,\cD)$ satisfies the basic assumptions listed in \citep[Assumption~6.5]{GoebSanf12}, so that system~\eqref{eq:genHybSys} has a well-defined solution $\xi$ in the space of hybrid arcs, for each hybrid signal $d$, but not necessarily unique.% see \citep{GoebSanf12} for the related solution concepts.

To study the stability notions of interest for hybrid arcs, we need some notation. A function $\alpha:\R_{\ge 0} \to \R_{\ge 0}$ is said to be of class $\cK$ if it is continuous, strictly increasing, and $\chi(0) = 0$. If $\alpha$ is also unbounded, then it is said to be of class $\cK_\infty$. A function $\beta: \R_{\ge 0} \times \R_{\ge 0} \to \R_{\ge 0}$ is said to be of class $\cK\cL$ if $\beta (\cdot,t)$ is of class $\cK$ for each $t \in \R_{\ge 0}$ and $\beta(r,t) \to 0$ as $t \to \infty$ for each fixed $r \in \R_{\ge 0}$; see \cite[Chapter~4]{Khalil02} for their use in formulation of common stability notions. In addition, we require class $\cK\cL\cL$ function: A function $\beta:\R_{\ge 0} \times \R_{\ge 0} \times \R_{\ge 0} \to \R_{\ge 0}$ is a class $\cK\cL\cL$ function if $\beta(\cdot,\cdot,j)$ is a class $\cK\cL$ function for each $j \ge 0$ and $\beta(\cdot,s,\cdot)$ is a class $\cK\cL$ function for each $s \ge 0$. For a compact set $\cA \subset \cX$ and $\xi \in \cX$, $\vert \xi \vert_{\cA} := \inf_{z\in\cA} \vert \xi - z \vert$, where $\vert \cdot \vert$ denotes the usual Euclidean norm. Following \citep{CaiTeel09}, for a hybrid signal $d$, we use the notation $\|d \|_{(t, j)}$ to denote the maximum between $\displaystyle \esssup_{(\hat t, \hat j+1) \not \in \dom d, \hat t+ \hat j \le t + j} \vert d(\hat t, \hat j) \vert$ and $\displaystyle\sup_{(\hat t, \hat j+1) \in \dom d, \hat t+ \hat j\le t+ j} \vert d(\hat t, \hat j) \vert $. For positive-valued functions $\alpha, \chi$ over $\R_{\ge 0}$, we also use the Landau-notation to write $\alpha(s) = o(\chi(s))$ as $s \to c$ if $\lim_{s\to c} \frac{\alpha(s)}{\chi(s)} = 0$; similarly, we write $\alpha(s) = O(\chi(s))$ as $s \to c$ if $\lim_{s\to c} \frac{\alpha(s)}{\chi(s)} \le M$ for some $M > 0$.

\subsection{Input-to-State Stability}
We recall basic definitions and Lyapunov characterizations of ISS for hybrid systems from \citep{CaiTeel09}.
\begin{defn}[Input-to-state stability (ISS)]
System \eqref{eq:genHybSys} is {\it ISS} with respect to a compact set $\cA \subset \cX$ if there exist functions $\gamma \in \cK$, and $\beta \in \cK\cL\cL$ such that
\begin{equation}\label{eq:defISS}
|\xi(t,j)|_{\cA} \leq \beta(|\xi(0,0)|_{\cA}, t,j) + \gamma \left( \|d\|_{(t,j)} \right),
\end{equation}
for every $(t,j) \in \dom \xi$.
\end{defn}

\begin{defn}
[ISS Lyapunov function] A smooth function $V: \cX \to \Rposi$ is an ISS-Lyapunov function of the hybrid system \eqref{eq:genHybSys} w.r.t. a compact set $\cA \subset \cX$ if the following hold:
\begin{itemize}[leftmargin=*]
\item  there exist $\ul\alpha, \ol\alpha \in \cKinfty$ such that
\begin{equation}\label{boundV} 
\ul\alpha(\vert \xi \vert_{\cA}) \leq V(\xi) \leq \ol\alpha(\vert \xi \vert_{\cA}), \quad \forall \, \xi \in \cC \cup \cD \cup \cG(\cD),
\end{equation}
%\begin{subequations}
\item there exist $\widehat\alpha \in \cK_\infty$ and $\widehat\gamma \in \cK$ such that
\begin{equation}\label{eq:defIssVFlow}
\vert\xi\vert_{\cA} \geq \widehat\gamma(|d|) \Rightarrow \langle \nabla V(\xi), f \rangle \leq -\widehat\alpha(\vert\xi\vert_{\cA}),
\end{equation}
holds for every $(\xi,d) \in \cC$ and $f \in \cF(\xi,d)$.
\item the functions $\widehat \alpha$ and $\widehat \gamma$ also satisfy
\begin{equation}
\vert\xi\vert_{\cA} \geq \widehat\gamma(|d|) \Rightarrow V(g) - V(\xi) \leq -\widehat\alpha(\vert\xi\vert_{\cA}),
\end{equation}
for every $(\xi,d) \in \cD$ and for every $g \in \cG(\xi, d)$.
%\end{subequations}
\end{itemize}
\end{defn}
%Another equivalent way to write the ISS Lyapunov functions is that \cite{Libe03}
%\begin{equation}\label{ISSpropfirst}
%\frac{\partial V}{\partial \xi} f(\xi,d) \leq -\alpha(|\xi|) + \chi(|d|), \quad \forall \xi,d
%\end{equation} 
%where $\alpha, \chi \in \cKinfty$.
The following result provides an alternate charaterization of ISS for system~\eqref{eq:genHybSys} by combining results given in \citep[Proposition~1]{CaiTeel09} and \citep[Theorem~1]{LibeShim15}.
\begin{prop}\label{hybridISSprop}
Consider system~\eqref{eq:genHybSys} and a compact set $\cA \subset \cX$. A differentiable function $V: \cX \to \Rposi$ satisfying \eqref{boundV} with $\ul \alpha, \ol \alpha \in \cKinfty$ is an ISS-Lyapunov function w.r.t. $\cA$ if and only if 
\begin{itemize}[leftmargin=*]
\item there exist $\alpha_\cC \in {\color{blue}\cK_\infty}$, $\gamma_\cC \in \cK$ and a continuous nonnegative function $\varrho:\R_{\ge 0} \to \R_{\ge 0}$ such that
\begin{subequations}
\begin{equation}\label{prop1Vflow} 
\langle \nabla V(\xi), f \rangle \leq -\alpha_\cC(\vert\xi\vert_{\cA})+ \varrho(\vert \xi \vert_{\cA}) \, \gamma_\cC(|d|)
\end{equation}
holds for each $(\xi,d) \in \cC$ and $f \in \cF(\xi,d)$,
\item there exist $\alpha_\cD \in \cKinfty$, $\gamma_\cD \in \cK$ that satisfy
\begin{equation}\label{prop1Vjump}
V(g) - V(\xi) \leq - \alpha_\cD(\vert\xi\vert_{\cA})+ \gamma_\cD(|d|)
\end{equation}
\end{subequations}
for each $(\xi,d) \in \cD$ and $g \in \cG(\xi, d)$,
\item the functions $\alpha_\cC,\varrho$ satisfy the {\it asymptotic ratio} condition
\begin{equation}\label{eq:ratioCond}
\limsup_{r \to \infty} \, \frac{\varrho(r)}{\alpha_\cC(r)} = 0.
\end{equation}
\end{itemize}
\end{prop}

\begin{remark}
The inequality \eqref{prop1Vflow} is different from the expression given in \citep[{\color{blue}Proposition~2.6}]{CaiTeel09}. It can be shown that \eqref{prop1Vflow} also implies \eqref{eq:defIssVFlow}. This implication is proved in a constructive manner, that is, the pair $(\widehat \alpha, \widehat \gamma)$ is constructed from the triplet $(\alpha_\cC,\varrho,\gamma_\cC)$, in \citep[Theorem~1]{LibeShim15} using the condition \eqref{eq:ratioCond}, which appears in Remark~1 of that paper.
\end{remark}

\subsection{Cascade Switched Systems}

The framework of \eqref{eq:genHybSys} is useful for modeling switched systems. We are interested in studying switched systems in cascade configuration which comprise a family of dynamical subsystems described by
\begin{subequations}\label{cascnonlinearsys}
\begin{align}
&\dot x = f_{c,p}(x,e), \label{stage1p} \\
&\dot e = f_{o,p}(e,d), \label{stage2p}
\end{align}
\end{subequations}
where $p$ belongs to a {\it finite} index set $\cP$. The vector fields $f_{c,p} : \R^{n_c} \times \R^{n_o} \to \R^{n_c}$ and $f_{o,p}:\R^{n_o}\times \R^{n_d} \to \R^{n_o}$ are assumed to be continuous for each $p \in \cP$. It is also assumed that $f_{c,p}(0,0) = 0$, and $f_{o,p} (0,0) = 0$, and the stability of the origin $\{0\} \in \R^{n_c + n_o}$ is the topic of interest in the sequel.
The switched system generated by \eqref{cascnonlinearsys} is 
\begin{subequations}\label{cascswisys}
\begin{align}
\dot x = f_{c,\sigma}(x,e)\\
\dot e = f_{o,\sigma}(e,d),
\end{align}
\end{subequations}
where $\sigma:\R_{\ge 0} \to \cP$ denotes the piecewise constant right-continuous switching signal. The function $\sigma$ changes its value at switching times which are denoted by $\{t_i\}_{i \in \N}$. At these switching times, we allow the state values to have jumps defined by the maps
\begin{subequations}\label{eq:jumpSwSys}
\begin{align}
x^+ &= g_c(x,e) \\%x(t_i^+) &= g_c(x(t_i), e(t_i), d(t_i))  \\
e^+ & = g_o(e,d),%e(t_i^+) &= g_o(x(t_i), e(t_i), d(t_i)).
\end{align}
\end{subequations}
so that $x(t_i^+) = \big(x(t_i)\big)^+$, and $e(t_i^+) = \big(e(t_i)\big)^+$ denote the value of the state variables just after the switching times.
We say that the switching signal $\sigma$ has an average dwell-time $\tau_a$, denoted $\sigma \in \cS_{\tau_a}$ if there exists $N_0 \ge 1$ such that for each $t > s \ge 0$, it holds that
\begin{equation}\label{defADT}
N_{\sigma(t, s)} \leq N_0 + \frac{t-s}{\tau_a}
\end{equation}
where $N_{\sigma(t, s)}$ is the number of switching in the interval $(s, t]$. The constant $N_0$ is called the {
\it chatter bound} giving the tolerance number of fast switchings. 

\paragraph{Problem~1}Given that each subsystem in \eqref{stage1p} (with $e$ as input) and \eqref{stage2p} (with $d$ as input) admits an ISS Lyapunov function w.r.t. the origin, how can we 
\begin{enumerate}
\item compute the lower bound on $\tau_a$, and
\item construct an ISS Lyapunov function for the hybrid system \eqref{cascswisys}-\eqref{eq:jumpSwSys},
\end{enumerate}
such that, for each $\sigma \in \cS_{\tau_a}$, we have
\[
\left\vert (x(t,j),e(t,j)) \right\vert \le \beta (\left\vert (x(0,0),e(0,0)) \right\vert, t, j) + \gamma \left( \|d\|_{(t,j)} \right)
\]
for some $\beta \in \cK\cL\cL$, and $\gamma \in \cK$.

\section{Stability of Cascade System}\label{sec:cascade}

To find a solution to the problem mentioned above, we proceed in several steps which allow us to arrive at the result given in Theorem~\ref{thm:mainISS}.

\subsection{Individual Lyapunov Functions}
The first step is to formally state the stability assumptions imposed on the dynamical subsystem \eqref{stage1p} and \eqref{stage2p} which are formally listed below:
\begin{lyap}[leftmargin=*]
\item\label{Voprop} For each $p \in \cP$, there exists a continuously differentiable function $V_{o,p} : \R^{n_o} \to \Rposi$, and there exist class $\cK_\infty$ functions $\overline{\alpha}_{o,p}$, $\underline{\alpha}_{o,p}$, $\alpha_{o,p}$ and $\gamma_{o,p}$ such that
\[
\underline{\alpha}_{o,p}(|e|) \leq V_{o,p}(e) \leq \overline{\alpha}_{o,p}(|e|)
\]
\[
\left\langle\frac{\partial V_{o,p}}{\partial e}, f_{o,p}(e,d)\right\rangle \leq - \alpha_{o,p}(V_{o,p}(e)) + \gamma_{o,p}(|d|)
\]
 hold for every $(e,d) \in \R^{n_o} \times \R^{n_d}$.
\end{lyap}
\begin{lyap}[leftmargin=*]
\item\label{Vcprop} For each $p \in \cP$, there exists a continuously differentiable function $V_{c,p} : \R^{n_c} \to \Rposi$, and there exist class $\cK_\infty$ functions $\overline{\alpha}_{c,p}$, $\underline{\alpha}_{c,p}$, $\alpha_{c,p}$, and $\gamma_{c,p}$ such that
\begin{subequations}
\[
\underline{\alpha}_{c,p}(|x|) \leq V_{c,p}(x) \leq \overline{\alpha}_{c,p}(|x|) 
\] 
\[
\left\langle\frac{\partial V_{c,p}}{\partial x} f_{c,p}(x,e) \right\rangle\leq - \alpha_{c,p}(V_{c,p}(x)) + \gamma_{c,p}(V_{o,p}(e)),
\]
\end{subequations}
hold for every $(x,e) \in \R^{n_c} \times \R^{n_o}$.
\end{lyap}

\begin{lyap}[leftmargin=*]
\item\label{Assumq} As $s \to 0^+$, we have $\gamma_{c,p}(s) = O(\alpha_{o,p}(s))$, that is, if we let
\begin{equation}\label{eq:defbarnu}
\ol \nu_p(s) := \frac{\gamma_{c,p} (s)}{\alpha_{o,p}(s)}, \quad \text{for }s > 0,
\end{equation}
then there exists a constant $M >0$, such that
\[
\lim_{s \to 0^+}  \ol \nu_p(s) \leq M.
\]
\end{lyap}

In addition, we introduce the following assumption\footnote{It is also possible to consider more general jump maps of the form $x^+ = g_c(x,e,d)$ and $e^+ = g_o(x,e,d)$, provided that the inequalities in \ref{assJump} take the form $| g_c(x,e,d)| \leq \widehat{\alpha}_c(|(x,e)|) + \widehat{\rho}_c(|d|)$ and $| g_o(x,e,d)| \leq \widehat{\alpha}_o(|(x,e)|) + \widehat{\rho}_o(|d|)$. The results of this paper would carry just by changing the map $\rho$ in \eqref{eq:defrho}.} on the jump maps introduced in \eqref{eq:jumpSwSys}.

\begin{assum}[leftmargin=*]
\item\label{assJump} For each $(x,e,d) \in \R^{n_c + n_o+n_d}$, the jump maps at switching times satisfy
\begin{align}
& | g_c(x,e)| \leq \widehat{\alpha}_c(|(x,e)|) \notag \\% + \widehat{\rho}_c(|d|) \notag \\
& | g_o(e,d)| \leq \widehat{\alpha}_o(|e|) + \widehat{\rho}_o(|d|) \notag 
\end{align}
for some class $\cK_\infty$ functions $\widehat{\alpha}_c$, $\widehat{\alpha}_o$, $\widehat{\rho}_c$, and $\widehat{\rho}_o$.
\end{assum}

Using the assumptions introduced above, it is possible to construct a candidate Lyapunov function $V_p(x,e)$ for each subsystem $p\in \cP$. This construction is primarily inspired from the work of \citep{TanwTeel15}.

\begin{remark}
The assumption~\ref{Assumq} is introduced to provide a construction of the candidate Lyapunov function $V_p$ explicitly in terms of $V_{o,p}$ and $V_{c,p}$. One can always modify the function $V_{c,p}$ to $\widehat V_{c,p}$ such that the resulting $\widehat\gamma_{c,p}$ satisfies \ref{Assumq}; this is a direct consequence of \citep[Theorem~1]{SontTeel95} as we can choose $\widehat\gamma_{c,p}(s) = O(\alpha_{o,p}(s)) \cdot \underline{\alpha}_{o,p}(s)$ for $s$ sufficiently small in the neighborhood of origin.
\end{remark}
%%%%%%%%%%%%%           PROP 2           %%%%%%%%%%%%%%%
\begin{prop}\label{basicVp}
Consider the family of dynamical subsystems \eqref{cascnonlinearsys} satisfying \ref{Voprop}, \ref{Vcprop}, and \ref{Assumq}, along with the jump dynamics \eqref{eq:jumpSwSys} satisfying \ref{assJump}. For each $p \in \cP$, introduce the continuously differentiable function $V_p$,
\begin{equation}\label{defVp}
V_p(x,e) := \int_0^{V_{o,p}(e)} \nu_p(s) \, ds + V_{c,p}(x),
\end{equation}
where $\nu_p:\R_{\ge 0} \to \R_{\ge 0}$ is a continuous and nondecreasing function with $\nu_p(s) \ge 4 \ol\nu_p(s)$, for each $s > 0$. It then holds that, for some $\ul\alpha_p, \ol\alpha_p \in \cK_\infty$,
\begin{equation}\label{eq:VpBound}
\ul\alpha_p(\vert(x,e)\vert) \le V_p(x,e) \le \ol\alpha_p(\vert(x,e)\vert), \quad \forall (x,e) \in \R^{n_c+n_o}.
\end{equation}
There also exist $\alpha_p,\gamma_p\in \cKinfty$ such that
\begin{equation}\label{eq:VpDecayBound}
\left\langle \nabla V_p(x,e), \begin{pmatrix} f_{c,p}(x,e)\\ f_{o,p}(e,d)\end{pmatrix}\right\rangle \le -\alpha_p(V_p(x,e)) + \gamma_p (|d|),
\end{equation}
for every $(x,e,d) \in \R^{n_c+n_o+n_d}$.
Moreover, there exist $\chi, \rho \in \cKinfty$ such that for each $(p,q) \in \cP \times \cP$, $q \neq p$,
\begin{equation}\label{eq:VpJumpBound}
V_q (x^+,e^+) \le \chi (V_p(x,e)) + \rho (|d|),
\end{equation}
for every $(x,e,d) \in \R^{n_c+n_o+n_d}$.
\end{prop}
\begin{proof}
Fix $p\in\cP$. Introduce the class $\cKinfty$ function $\ell_p: \Rposi \to \Rposi $ as follows: 
\begin{equation}\label{defl}
\ell_p(s) = \int_{0}^{s} \nu_p(r) dr
\end{equation}
where $\nu_p$ was introduced in \eqref{defVp}, so that
\[
V_p(x,e) = (\ell_p \circ V_{o,p})(e) + V_{c,p}(x).
\]
The bound \eqref{eq:VpBound} is seen to hold since $\ell_p$ is a class $\cK_\infty$ function.
Using \ref{Voprop}, we now obtain
\begin{multline}\label{eq:dlVoTemp}
\left\langle \nabla (\ell_p \circ V_{o,p})(e), f_{o,p}(e,d) \right\rangle \\ 
\le \nu_p(V_{o,p}(e)) ( - \alpha_{o,p}(V_{o,p}(e)) + \gamma_{o,p}(|d|)).
\end{multline}
To analyze the right-hand side of \eqref{eq:dlVoTemp}, first consider the case where $\gamma_{o,p}(\vert d \vert) \le \frac 12 \alpha_{o,p}(V_{o,p}(e))$, so that
\[
\left\langle \nabla (\ell_p \circ V_{o,p})(e), f_{o,p}(e,d) \right\rangle \le -\frac 12 \nu_p(V_{o,p}(e))\alpha_{o,p}(V_{o,p}(e));
\]
else, by introducing $\theta_p(s):=\alpha_{o,p}^{-1}(2\gamma_{o,p} (s))$,
\begin{align*}
\frac{1}{2} \alpha_{o,p}(V_{o,p}(e)) \le \gamma_{o,p}(\vert d \vert) \Leftrightarrow V_{o,p}(e) & \le \alpha_{o,p}^{-1}(2\gamma_{o,p}(\vert d \vert)) \\
& = \theta_p(\vert d \vert)
\end{align*}
so that $\nu_p(V_{o,p}(e)) \le \nu_p(\theta_{p}(\vert d \vert))$, because $\nu_p$ is by construction nondecreasing, and
\begin{multline*}
 \left\langle \nabla (\ell_p \circ V_{o,p})(e), f_{o,p}(e,d) \right\rangle \le \\
- \nu_p(V_{o,p}(e))\alpha_{o,p}(V_{o,p}(e)) + \nu_p(\theta_p(\vert d \vert))\gamma_{o,p}(\vert d \vert).
\end{multline*}
From these two cases, the inequality \eqref{eq:dlVoTemp} results in
\begin{multline}\label{eq:dlVo}
 \left\langle \nabla (\ell_p \circ V_{o,p})(e), f_{o,p}(e,d) \right\rangle \le \\
-\frac 12 \nu_p(V_{o,p}(e))\alpha_{o,p}(V_{o,p}(e)) + \nu_p(\theta_p(\vert d \vert))\gamma_{o,p}(\vert d \vert).
\end{multline}

Using \ref{Vcprop}, \eqref{eq:dlVo}, and the fact that $\nu_p(s) \ge 4 \ol \nu_p(s)$, for each $s > 0$ with $\ol \nu_p$ given in \eqref{eq:defbarnu}, we can now derive \eqref{eq:VpDecayBound} as follows:
\begin{align*}
 & \left\langle \nabla V_p(x,e), \begin{pmatrix} f_{c,p}(x,e)\\ f_{o,p}(e,d)\end{pmatrix}\right\rangle \\
 & \quad \le -\frac 12 \nu_p(V_{o,p}(e))\alpha_{o,p}(V_{o,p}(e)) + \nu_p(\theta_p(\vert d \vert))\gamma_{o,p}(\vert d \vert) \\
 & \quad \quad - \alpha_{c,p}(V_{c,p}(x)) + \gamma_{c,p}(V_{o,p}(e))\\
& \quad \le   - \gamma_{c,p}(V_{o,p}(e))) - \alpha_{c,p}(V_{c,p}(x)) + \nu_p(\theta_p(|d|))  \gamma_{o,p}(|d|) \\
& \quad \le -\alpha_p(V_p(x,e)) + \gamma_p(\vert d \vert),
\end{align*}
where we used the definitions $\gamma_p(s) := \nu_p(\theta_p(s)) \gamma_{o,p}(s)$,
\begin{equation}\label{defalphapinput}
\alpha_p(s) := \min \left\{ \alpha_{c,p} \left(\demi s\right),  \gamma_{c,p} \left( \demi \ell_p\inv (s)\right) \right\},
\end{equation}
and the triangle-inequality type result from \citep[Lemma~10]{Kell14} to derive the last inequality.

Next, to derive \eqref{eq:VpJumpBound}, we observe that
\[
V_q(x^+,e^+) = (\ell_q \circ V_{o,q})(g_o(e,d)) + V_{c,q}(g_c(x,e)).
\]
Using \ref{assJump}, it then follows that
\begin{align}
V_q(x^+,e^+) &\leq \ell_q \circ \overline{\alpha}_{o,q}(|g_o(e,d)|) + \overline{\alpha}_{c,q}(| g_c(x,e)|) \notag \\
& \leq \ell_q \circ \overline{\alpha}_{o,q} \left(2 \widehat{\alpha}_o(|(x,e)|) \right) +  \ell_q \circ \overline{\alpha}_{o,q} (2 \widehat{\rho}_o(|d|) \notag \\
&\quad + \overline{\alpha}_{c,q}(\widehat{\alpha}_c(|(x,e)|)) \notag \\ %+ \overline{\alpha}_{c,q}(2 \widehat{\rho}_c(|d|))\notag \\
&\leq \chi(V_p(x,e)) + \rho (|d|)
\end{align}
where, recalling \eqref{eq:VpBound}, we used the definitions
\begin{equation}\label{eq:defchi}
\chi(s) :=  \max_{p,q \in \cP} \left\{ \ell_q \circ \overline{\alpha}_{o,q} \left(2 \widehat{\alpha}_o \circ \ul\alpha_p\inv (s)\right) + \overline{\alpha}_{c,q}(\widehat{\alpha}_c \circ \ul\alpha_p\inv (s)) \right\}
\end{equation} and 
\begin{equation}\label{eq:defrho}
\rho(s) := \max_{q \in \cP} \left\{ \ell_q \circ \overline{\alpha}_{o,q} (2  \widehat{\rho}_o(s)) \right\}, %+ \overline{\alpha}_{c,q}(2 \widehat{\rho}_c(s))
\end{equation}
which establishes the desired bound since both functions are class $\cK_\infty$.
%
%Assume that during the switching from subsystem $p$ to $q \in \cP$, the jump map for the state $x$ and $e$ can be expressed by the function $h_c$ and $h_o$, respectively, i.e.
%\begin{align}
%&x^+ = h_c(x,e,d) \notag \\
%&e^+ = h_o(x,e,d), \notag
%\end{align}
%where we assume that the function $h_o$ and $h_c$ satisfy:
%Thus, the multiple Lyapunov function $V_q$ can be written as
%\[
%V_q(x,e) = V_{p^+}(x^+, e^+) = l \circ V_{o,q}(e^+) + V_{c,q}(x^+)
%\]
\end{proof}
% \alpha_{1,p} |(x,e,d)| & = l \circ \underline{\alpha}_{o,p}(| h_o(x,e,d)|) + \underline{\alpha}_{c,p}(| h_c(x,e,d)|)  \notag \\
% \alpha_{2,p} |(x,e,d)| & =  l \circ \overline{\alpha}_{o,p}(| h_o(x,e,d)|) + \overline{\alpha}_{c,p}(| h_c(x,e,d)|). \notag
\subsection{Stability of Overall Hybrid System}\label{sec:plainHybSys}
We now use the result of Proposition~\ref{basicVp} to compute lower bounds on the dwell-time which result in cascade switched system being globally ISS.
To do so, we find it convenient to express the switched system in the framework of the hybrid system adopted in \citep{GoebSanf12}.
This is done by introducing the augmented state variable $\xi := \left( x, e, p, \tau \right) \in \cX:= \R^{n_c+n_o} \times \cP \times [0,N_0]$, where $p$ is a discrete variable denoting a subsystem, and $\tau$ plays the role of a scaled timer. The hybrid model capturing the dynamics of the switched system driven by an external disturbance $d \in \R^{\ol d}$, and where the switching signals have an average dwell-time $\tau_a$, is
\begin{subequations}\label{eq:smallHybSys}
\begin{align}
&(\xi,d) \in \cC:
\begin{cases}
\dot x = f_{c,p}(x, e)\\
\dot e = f_{o,p}(e, d)\\
\dot p = 0 \\
\dot \tau \in [0, \frac{1}{\tau_a}]\\
\end{cases} \label{eq:smallHybSysa} \\
&(\xi,d) \in \cD:
\begin{cases}
x^+ = g_c(x,e)  \\
e^+ = g_o(e,d)  \\
p^+ \in \cP \setminus \{p\} \\
\tau^+ = \tau -1 
\end{cases} \label{eq:smallHybSysb}
\end{align}
\end{subequations}
where the flow set $\cC:=\cX \times \R^{\ol d}$, and the jump set $\cD:=\R^{n_c+n_o} \times \cP \times [1,N_0] \times \R^{\ol d}$. We denote the set-valued mapping on the right-hand side of \eqref{eq:smallHybSysa} by $\cF(\xi,d)$, and the mapping on the right-hand side of \eqref{eq:smallHybSysb} is denoted by $\cG(\xi,d)$.
We are interested in studying the ISS property of the system~\eqref{eq:smallHybSys} (driven by the disturbance $d$) with respect to the compact set
\begin{equation}\label{eq:defA0}
\cA_0 := \{0\}^{n_c+n_o} \times \cP \times [0,N_0]
\end{equation}
by finding an appropriate ISS Lyapunov function. To do so, we introduce the function $\varphi:\R_{\ge 0} \to \R_{\ge 0}$ defined as
\begin{equation}\label{eq:defPhi}
\varphi(s) = \begin{cases} \exp \left( \int_{1}^{s} \frac{2{\color{blue}c_0}}{\psi(r)} \, dr \right), & s > 0 \\ 0, & s = 0 \end{cases}
\end{equation}
where $\psi:\R_{\ge 0} \to \R_{\ge 0}$ is a continuously differentiable class $\cK_\infty$ function, with $\psi'(0) = 0$, and
\begin{align*}
& \psi(s) \le \min\{c_0s, \alpha_p(s) \, \vert \, p \in \cP\}, \quad {\color{blue} s \ge 0},\\
{\color{blue}\text{\sout{and}}} \quad & {\color{blue}\hbox{\sout{$\psi(s) \le \min\{\alpha_p(s) \, \vert \, p \in \cP\}, \quad s \ge c_1$}}},
\end{align*}
for some $c_0 > 0$. We recall that $\alpha_p \in \cK_\infty$ were introduced as the decay function for $V_p$ in \eqref{eq:VpDecayBound}.
The function $\varphi$ is now used in the following result:

%%%%%%%%%%%%%%%%%%%       %%%%%%%%%%%%%%%%%%
%%%%%%%%%%%%%%%%%%%%%%%%%%%%%%%%%%%%%%%%%%%%%%%%%

%%%% %%%%%%  THEOREM  1     %%%%%%%%%%%%%%%%%%
\begin{thm}\label{thm:mainISS}
Consider system \eqref{eq:smallHybSys} and suppose that \ref{Voprop}, \ref{Vcprop}, \ref{Assumq}, \ref{assJump} hold. Let $\chi \in \cK_\infty$ be as in \eqref{eq:VpJumpBound}.
If, for some $\eps>0$, the average dwell-time $\tau_a$ satisfies
\begin{equation}\label{condADT}
\tau_a > \zeta^* :=  \sup_{s \ge 0} \, \int_{s }^{(1+\eps)\chi(s)} \frac{1}{\psi(r)} \, dr,
\end{equation}
then for each $\tau_a > \zeta > \zeta^*$, 
\begin{equation}\label{eq:defWp}
W (x,e,p,\tau) := \exp(2 \, {\color{blue}c_0} \, \zeta\tau)\varphi(V_p(x,e)), 
\end{equation}
is an ISS Lyapunov function for the hybrid system \eqref{eq:smallHybSys} w.r.t. the compact set $\cA_0$, and input $d$.
\end{thm}

\begin{remark}
To gain an insight about the constraints imposed by the stability condition \eqref{condADT} on the system structure, we study particular instances where $\alpha(s) := \min\{\alpha_p(s)\, \vert \, p \in \cP\}$ exhibits linear, super-linear and sub-linear growth. It can be seen that if the jump map $\chi$ does not grow too fast compared to $\alpha(s)$, then $\zeta^*$ in \eqref{condADT} is finite. For the sake of simplicity, let $\alpha(s) := a s^k$, with $a,k > 0$, and choose $c_0 = a$ in the definition of $\psi$.

\begin{itemize}[leftmargin=*]
\item {\it Linear decay:} We first consider the case $k = 1$, so that $\alpha(s) = a s$, and we let $\psi (s) = a s$, for $s \ge 0$. This gives $\zeta^* = \sup_{s\ge 0}\frac{1}{a} \log{\frac{(1+\eps)\chi(s)}{s}}$, which is finite if $\lim_{s\to\infty} \chi (s) = O(s)$, and $\lim_{s \to 0} \chi (s) = O(s)$. If $\chi(s) = \mu s$, then $\zeta^* = \frac{1}{a}\log((1+\eps)\mu)$, which resembles the bound given in \citep[Chapter~3]{Libe03} by taking $\eps > 0$ arbitrarily small.

\item {\it Super-linear decay:} Next consider the case where $\alpha(s) = a s^k$ with $k > 1$. Choose $c_1 =1$, then we can let $\psi(s) = a s^k$, for every {\color{blue}$s \in [0,c_1]$. Thus, $\zeta^*$ is the maximum between $\frac{1}{k-1}\sup_{s\in [0,c_1]} \left[\frac{1}{as^{k-1}} - \frac{1}{a\chi(s)^{k-1}}\right]$ and $\sup_{s\ge 1}\frac{1}{a} \log{\frac{(1+\eps)\chi(s)}{s}}$}. With $\chi \in \cK_\infty$, it is seen that $\zeta^*$ is finite and positive if $\lim_{s\to 0^+} \left[\frac{1}{as^{k-1}} - \frac{1}{a\chi(s)^{k-1}}\right] \le 0$, which holds if $\chi(s) = o(s)$ when $ s \to 0$, and {\color{blue}$\lim_{s\to\infty} \chi(s) = O(s)$.}

\item {\it Sub-linear decay:} Lastly, consider the case where $\alpha(s) = a s^k$ with $k < 1$. Choose $c_1 =1$, then there exist $\ul c < 1$, $\ol c >1$ and a continuously differentiable function $\psi$ such that  $\psi(s) = as$, $s\in [0,\ul c]$, and $\psi(s) = a s^k$ for $s \ge \ol c$. With this choice of $\psi$, the lower bound $\zeta^*$ is a finite positive scalar, if $\lim_{s\to 0^+} \chi(s) = O(s)$, and $\lim_{s\to\infty} \chi(s) = o(s)$.
\end{itemize}
\end{remark}

\begin{remark}\label{rem:phi}
A function similar to $\varphi$ defined in \eqref{eq:defPhi} also appears in \citep{PralWang96} to transform nonlinear decay rates to linear ones in inequalities associated with Lyapunov functions, while keeping the modified Lyapunov function differentiable. In the proof of Theorem~\ref{thm:mainISS}, this function serves the same purpose. Here, the construction of $\psi$ is modified slightly. %\sout{because the differentiability of $\varphi$ is only a concern at the origin. If one keeps the construction of $\psi$ as in \text{\citep{PralWang96}}, with $\psi = \min\{s,\alpha(s)\}$, $\alpha(s) = \min_{p \in \cP} \{\alpha_p(s)\}$, then $\zeta^* \ge {\color{blue}\sup_{s\ge 0} }\ln \frac{\chi(s)}{s}$, and one is restricted to consider jump maps $\chi$ which grow at most linearly away from the origin.}
\end{remark}

%%%%%%%%%%%%%%%      PROOF            %%%%%%%%
\begin{proof}[Proof of Theorem~\ref{thm:mainISS}]
The proof is based on showing that $W$ satisfies the conditions listed in Proposition~\ref{hybridISSprop}. It is seen that $\varphi$ is differentiable (away from the origin) on $\R_{>0}$, and it is shown in \cite[Lemma~12]{PralWang96} that the function $\varphi$ is also continuously differentiable in the neighborhood of the origin with $\varphi'(0) = 0$. Therefore, $W$ is also continuously differentiable.

Since $\varphi$ is a class $\cK_\infty$ function, one can easily verify that
\[
\ul \alpha(\vert \xi \vert_{\cA_0}) \le W(\xi) \le \ol\alpha(\vert \xi \vert_{\cA_0})
\]
for some functions $\ul \alpha$, $\ol \alpha$ of class $\cK_\infty$. From the definition of $\cA_0$ and the assumption that $f_{c,p}(0,0) = f_{o,p}(0,0) = 0$, $p \in \cP$, it immediately follows that for each $\xi \in \cA_0$ such that $(\xi,0) \in \cD$, we have $\cG(\xi,0) \in \cA_0$. Also, along any continuous motion resulting from \eqref{eq:smallHybSysa}, with initial condition $\xi \in \cA_0$ satisfying $(\xi,0)\in \cC$, the system trajectory stays within $\cA_0$. Hence, $\cA_0$ is forward invariant with $d = 0$.

Let $f$ be an element of $\cF(\xi,d)$. When $(\xi,d) \in \cC$, it follows from \eqref{eq:VpDecayBound} that 
\begin{multline*}
\left\langle \nabla W(\xi), f \right\rangle  \le \frac{2 {\color{blue}c_0}\zeta \exp({2 {\color{blue}c_0} \zeta \tau})}{\tau_a} \varphi(V_p(x,e))\\ 
 + \exp(2 {\color{blue}c_0} \zeta \tau) \varphi(V_p(x,e)) \left[-\frac{2{\color{blue}c_0}\alpha_p(V_p(x,e))} {\psi(V_p)} + \frac{2{\color{blue}c_0} \gamma_p(\vert d \vert)} {\psi(V_p)} \right]
 \end{multline*}
 and hence
 \[
 \left\langle \nabla W(\xi), f \right\rangle  \le
W(\xi) \left[ \frac{2 {\color{blue}c_0} \zeta}{\tau_a} - \frac{2 {\color{blue}c_0} \alpha_p (V_p)} {\psi(V_p) } + \frac{2 {\color{blue}c_0} \gamma_p(|d|)}{\psi(V_p)}\right].
\]
Since $\psi(s) \le \alpha_p(s)$ by construction, and $\zeta$ is chosen to satisfy $\tau_a >\zeta$, we get an $a:= 2{\color{blue}c_0}(1 - \zeta/\tau_a) > 0$ such that
\[
\left\langle \nabla W(\xi), f \right\rangle  \le - aW(\xi) + \frac{2 {\color{blue}c_0} W(\xi)}{\psi(V_p(x,e))} \gamma_p(|d|), \quad \xi \in \cC
\]
which is of the same form\footnote{The function $\frac{2W(\xi)}{\psi(V_p(x,e))}$ is obviously nonnegative. The continuity follows from recalling that $\varphi$ is continuously differentiable with $\varphi'(0) = 0$, and observing that $\frac{2W(\xi)}{\psi(V_p(x,e))} = \exp{(2\zeta\tau)}\frac{\varphi(V_p(x,e))}{\psi(V_p(x,e))} = \exp{(2\zeta\tau)} \varphi'(V_p(x,e))$.} as \eqref{prop1Vflow}. It is readily checked that the asymptotic ratio condition \eqref{eq:ratioCond} holds since
\[
\lim_{\vert(x,e)\vert \to \infty}  \frac{1}{\psi(V_p(x,e))} = 0.
\]

The next step is to show that \eqref{prop1Vjump} holds under condition \eqref{condADT} for the jump maps \eqref{eq:smallHybSysb}. Let $g$ denote an element of $\cG(\xi,d)$. It is seen that for each $(\xi,d) \in \cD$, that is, whenever $\tau \in [1,N_0]$, it follows from \eqref{eq:VpJumpBound} that
\begin{align}
\max_{g \in \cG} W^+ &= \ \max_{g \in \cG} \ \exp(2 {\color{blue}c_0} \zeta\tau^+) \varphi( V_{p^+}(x^+,e^+)) \notag \\
& \le \ \exp(2 {\color{blue}c_0} \zeta \tau - 2 {\color{blue}c_0}\zeta)  \varphi \left( \chi (V_p(x,e)) + \rho(|d|)\right), %\notag
\end{align}
where $\chi$ and $\rho$ were introduced in \eqref{eq:defchi} and \eqref{eq:defrho}.
Take an $\eps > 0$ for which \eqref{condADT} holds, then it follows from~\cite[Lemma~10]{Kell14} that: 
\begin{multline*}
\varphi \left( \chi (V_p(x,e)) + \rho(|d|)\right) \le \varphi((1+\eps)\chi(V_p(x,e))) \\
+ \varphi \left(\frac{1+\eps}{\eps}\rho(|d|)\right).
\end{multline*}
For $(x,e) \neq 0$, the bound on the first term on the right-hand side is given by
\begin{align*}
&\varphi((1+\eps)\chi(V_p(x,e))) = \exp\left(\int_{1}^{(1+\eps)\chi(V_p(x,e))} \frac{2 {\color{blue}c_0} dr}{\psi(r)} \right)\\
&\quad = \exp\left(\int_{V_p(x,e)}^{(1+\eps)\chi(V_p(x,e))} \frac{2 {\color{blue}c_0} dr}{\psi(r)}\right) \cdot \exp\left( \int_{1}^{V_p(x,e)} \frac{2 {\color{blue}c_0} dr}{\psi(r)}\right)\\
& \quad \le \exp(2 {\color{blue}c_0} \zeta^*) \, \varphi (V_p(x,e))
\end{align*}
where $\zeta^*$ is defined as in \eqref{condADT}. Letting
\begin{equation}\label{eq:defrhotilde}
\widetilde{\rho}(s) :=\exp{(2 {\color{blue}c_0} \zeta N_0 - 2 {\color{blue}c_0} \zeta)} \ \varphi \left(\frac{1+\eps}{\eps}\rho(s)\right),
\end{equation}
and noting that, for $\tau \in [0,N_0]$,
\[
\widetilde \rho(s) \ge  \exp{(2 {\color{blue}c_0} \zeta \tau -2 {\color{blue}c_0} \zeta)} \ \varphi \left(\frac{1+\eps}{\eps}\rho(s)\right),
\]
we obtain
\begin{align*}
\max_{g \in \cG} W^+ & \leq \exp{(2 {\color{blue}c_0} \zeta^*-2 {\color{blue}c_0} \zeta)} \exp{(2{\color{blue}c_0}  \zeta \tau)}\varphi(V_p(x,e)) + \widetilde{\rho}(|d|) \\
& \leq \exp(2 {\color{blue}c_0} \zeta^* - 2 {\color{blue}c_0} \zeta) W (\xi) + \widetilde{\rho}(|d|).
\end{align*} 
Having chosen $\zeta$ such that $\zeta^*-\zeta < 0$, we see that \eqref{prop1Vjump} holds.
The result of Proposition~\ref{hybridISSprop} thus ensures that $W$ is an ISS Lyapunov function for \eqref{eq:smallHybSys} w.r.t. the set $\cA_0$, with input $d$.
\end{proof}

%%%%%%%%%%%%%%         LINEAR EXAMPLE          %%%%%%%%%%%%
\subsection{Linear Case}
We use the following linear example to illustrate the Theorem~\ref{thm:mainISS}.
Assume that the system \eqref{cascnonlinearsys} is linear
\begin{subequations}\label{casclinearsys}
\begin{align}
\dot x = A_px + B_p e \label{linearsysx}  \\
\dot e = F_p e + G_p d \label{linearsyse} 
\end{align}
\end{subequations}
with the matrices $A_p,F_p$ being Hurwitz.
For the sake of simplicity, we assume that the states $(x,e)$ do not go through any jump dynamics, and remain unchanged at switching instances. This way, we let the maps in \eqref{eq:jumpSwSys} be
\[
g_c(x,e) = x, \quad \text { and } \quad g_o(e,d) = e.
\]
We can now choose quadratic Lyapunov functions to satisfy \ref{Voprop} and \ref{Vcprop}. This is done by computing symmetric positive definite matrices $P_{o,p}, P_{c,p} > 0$ such that, for each $p \in \cP$,
\begin{align*}
F_p^\top P_{o,p} + P_{c,p} F_p & \le -Q_{o,p} \\
A_p^\top P_{c,p} + P_{c,p} A_p & \le -Q_{c,p}
\end{align*}
for some symmetric positive definite matrices $Q_{o,p}, Q_{c,p} > 0$. By letting
\[
V_{c,p}(x) = x^\top P_{c,p} x, \quad \text{ and } \quad \text V_{o,p} = e^\top P_{o,p} \, e
\]
we get
\[
\aoplowbar |x|^2 \leq V_{o,p}(x) \leq \aopupbar |x|^2
\]
with $\aoplowbar = \lambda_{\min}(P_{o,p})$, and $\aopupbar = \lambda_{\max}(P_{o,p})$.
Similarly, it holds that 
\[
\acplowbar |x|^2 \leq V_{c,p}(x) \leq \acpupbar |x|^2
\]
with $\acplowbar = \lambda_{\min}(P_{c,p})$, and $\acpupbar = \lambda_{\max}(P_{c,p})$.
It can be readily shown that
\[
\left\langle\nabla V_{o,p}, F_px + G_p e\right\rangle \le - a_{o,p} V_{o,p}(x) + \gmaopupbar |d|^2
\]
by letting
\[
a_{o,p} = \frac{\lambda_{\min}(Q_{o,p})}{2\lambda_{\max}(P_{o,p})} \quad \text{ and } \quad \ol\gamma_{o,p} = 2 \frac{\| P_{o,p}G_{p} \|^2}{\lambda_{\min}(Q_{o,p})}
\]
and likewise
\[
\left\langle\nabla V_{c,p}, A_px + B_p e\right\rangle \le - a_{c,p} V_{c,p}(x) + \gmacpupbar V_{o,p}(e)
\]
with
\[
a_{c,p} = \frac{\lambda_{\min}(Q_{c,p})}{2\lambda_{\max}(P_{c,p})} \quad \text{ and } \quad \gmacpupbar = 2 \frac{\| P_{c,p} B_{p} \|^2}{\lambda_{\min}(Q_{c,p})}.
\]
%\begin{align}
%\dot V_{c,p}(x) &=  x^T(A_p^T P_{c,p} + P_{c,p}A_p) x + e^T B_p^T P_{c,p}x + x^TP_{c,p}B_p e \notag \\
%& = - x^T Q_{c,p} x + 2 \norm{P_{c,p}} \norm{B_p} |x| |e|, \notag
%\end{align}
%where $ Q_{c,p} = -\left( A_p^T P_{c,p} + P_{c,p}A_p \right)$ is a positive definite matrice. It follows from the Young's Inequality that 
%\[
%\dot V_{c,p} \leq - a_{c,p} (1-\lambda_c) |x|^2 + \frac{b^2_{c,p}}{\lambda_c a_{c,p}} |e|^2
%\]
%where $a_{c,p} = \lambda_{min(Q_{c,p})}$, $b_{c,p} = \norm{P_{c,p}} \norm{B_p}$. If we choose $1- \lambda_c = \norm{P_{c,p}}$, the $\dot V_{c,p}$ can be written as
%\[
%\dot V_{c,p}(x) \leq - a_{c,p} V_{c,p}(x) + \gmacpupbar |e|^2,
%\]
%where $\gmacpupbar = \frac{b^2_{c,p}}{(1 - \norm{P_{c,p}}) a_{c,p}}$. 
%Using the same way of reasoning, the Lyapunov function and conditions of the system~\eqref{linearsyse} can be written as follows. 
%\begin{align}
%&V_{o,p}(e) = e^T P_{o,p} e \notag \\
%& \aoplowbar |e|^2 \leq V_{o,p}(e) \leq \aopupbar |e|^2 \notag \\
%& \dot V_{o,p}(e) \leq - a_{o,p} V_{o,p}(e) + \gmaopupbar |d|^2, \notag
%\end{align}
%where matrice $P_{c,p}$ is positive definite and symmetric and 
%\begin{align}
%\begin{cases}
%\aoplowbar &= \lambda_{min(P_{o,p})} \notag \\
%\aopupbar &= \lambda_{max(P_{o,p})} \notag \\
%a_{o,p} & =  \lambda_{min(Q_{o,p})}, \quad Q_{o,p} = - \left( F_p^T P_{o,p} + P_{o,p}F_p \right) \notag \\
%\gmaopupbar &=  \frac{b^2_{o,p}}{(1 - \norm{P_{o,p}}) a_{o,p}}, \quad b_{o,p} = \norm{P_{o,p}} \norm{G_p} \notag
%\end{cases}
%\end{align}
For each $p \in \cP$, the function $\ol\nu_p(s)$ in \eqref{eq:defbarnu} turns out to be a constant as
\[
\ol\nu_p(s) = \frac{\ol \gamma_{c,p} \, s}{a_{o,p} \, s} = \frac{\ol \gamma_{c,p}}{a_{o,p}} =: \ol\nu_p.
\]
Thus, we can choose the Lyapunov function in \eqref{defVp} to be
\[
V_p(x,e) = 4 \, \ol\nu_p V_{o,p}(e) + V_{c,p} (x)
\]
which leads to
\[
\left\langle \nabla V_p(x,e), \begin{pmatrix} A_p x + B_p e \\ F_p e + G_p d \end{pmatrix}\right\rangle \le -a_p \, V_p(x,e) + \ol \gamma_p |d|^2
\]
in which
\begin{align*}
&a_p = \min \left\{a_{c,p}, 0.75 \, a_{o,p} \right\},\\
&\gmaupbar = 4\ol\nu_p \gmaopupbar.
\end{align*}
For the given jump maps at switching times, the maps in \eqref{eq:VpJumpBound} can be chosen such that $\rho \equiv 0$, and 
\begin{equation}\label{eq:defchibar}
\chi(s) = \overline{\chi} s, \quad  \overline{\chi} = \max_{p,q \in \cP} \left\{ \frac{\ol \nu_q}{\ol \nu_p} \frac{\lambda_{\max}(P_{o,q})}{\lambda_{\min}(P_{o,p})}, \frac{\lambda_{\max}(P_{c,q})}{\lambda_{\min}(P_{c,p})} \right\}.
\end{equation}
To construct the Lyapunov function $W$ in \eqref{eq:defWp}, we let
\[
a:= \min_{p\in\cP} \{a_p\}
\] 
so that $\psi(s) = a s$ satisfies the desired conditions with $c_0 = a$ and $c_1 > 0$ arbitrary. We now choose $W$ such that $W(x,e,p,\tau) = 0$ if $(x,e) = 0$, and for $(x,e) \neq 0$,
\begin{align*}
W(x,e,p,\tau) &= \exp(2 {\color{blue}a} \zeta \tau) \exp \left( \int_{1}^{V_p(x,e)} \frac{2 {\color{blue}a}\,dr}{a\, r} \right) \\
& = \exp ({2{\color{blue}a}\zeta \tau}) V_p^{{\color{blue}2}}(x,e).
\end{align*}
This construction leads to the following result:
\begin{cor}
The switched linear system \eqref{cascswisys} with subsystems described by \eqref{casclinearsys} is input-to-state stable with respect to the origin and disturbance $d$ if 
\[
\tau_a > \frac{1}{a} \ln(\overline \chi),
\]
where $a = \min_{p \in \cP}\{a_p\}$, and $\ol \chi$ is given in \eqref{eq:defchibar}.
\end{cor}
%When the overall state of the system $\xi \in \cC$, we can compute the derivative of $W_p$ as:
%\begin{align}
%\dot W_p &= \mu e^{\mu \tau} e^{-c} \dot \tau V_p + e^{\mu \tau} e^{-c} V_p \dot V_p \notag \\
%&\leq e^{-c}  \left[ \mu \dot \tau - \frac{a_p V_p} {c V_p } + \frac{\gmaupbar |d|^2}{c V_p} \right] W_p. \notag \\
%& \leq e^{-c} \left[ \mu \dot \tau - 1 + \frac{\gmaupbar |d|^2}{c V_p} \right] W_p.
%\end{align}

%%%%%%%%%%%%%%%%%%%%%%%%%%%%%
\section{Output Feedback Stabilization with Sampling}\label{sec:sampling}
%\subsection{Problem setup}
In this section, we will use the theoretical results from the previous section to study the problem of feedback stabilization with dynamic output feedback for switched nonlinear systems with time-sampled measurements. Our starting point is a nominal setup where an output feedback controller is already designed for each subsystem. We assume this controller to be robust with respect to measurement errors, and this property is formalized using ISS notion. Next, we implement these controllers when the output measurements sent by the plant, and the control inputs sent by the controller are time-sampled and subjected to zero-order hold between two updates. Our goal is to design the sampling algorithms for control signals and measurements such that the resulting closed-loop system is asymptotically stable. Inspired by the work of \citep{TanwTeel15}, we introduce the dynamic filters and event-based update rules to design the sampling algorithms. 

\subsection{Problem Setup}
The switched nonlinear plant which we want to stabilize, is described by
\begin{subequations}\label{cascadeFam3}
\begin{align}
&\dot x = f_{c,\sigma} (x, u)  \label{eq:planta}\\
& y = h_\sigma(x),
\end{align}
\end{subequations}
and the corresponding controller is described by
\begin{subequations}\label{samplingcontroller}
\begin{align}
&\dot z = f_{o,\sigma}(z, u, y) \label{eq:obsa}\\
& u = k_\sigma(z).\label{eq:obsb}
\end{align}
\end{subequations}
Here, for a finite index set $\cP$, $\sigma:\R_{\ge 0} \to \cP$ is the switching signal, and we emphasize that the controller is driven by the same switching signal as the plant. For each $p \in \cP$, the mappings $f_{c,p}:\R^n \times \R^{n_u} \to \R^n$, $h_p:\R^{n} \to \R^{n_y}$, $f_{o,p}:\R^n \times \R^{n_u} \times \R^{n_y} \to \R^n$, and $k_p:\R^n \to \R^{n_u}$ are assumed to be continuous on their respective domains. The underlying working principle behind the controller \eqref{samplingcontroller} is that the variable $z$ in \eqref{eq:obsa} acts as a full-state estimate for the variable $x$ governed by \eqref{eq:planta}. The dynamics of the estimation error, denoted $e:=z-x$, are assumed to be ISS with respect to measurement errors, as described in \ref{spvop} below. Also, the static feedback controller $k_p$ in \eqref{eq:obsb}, $p \in \cP$, has the property that the corresponding subsystem $\dot x = f_{c,p}(x,k_p(x))$ is ISS with respect to errors in measurement of $x$, as stated in \ref{spvcp}. With these properties, we can design sampling algorithms for output $y$ and input $u$, where the difference between the last updated value of the output (resp.~input) and its current value acts as an error in measurement. 

In our setup of sampling, we use a zero-order-hold between sampling times so that the inputs and outputs stay constant between two successive updates. Thus, we introduce two additional states in our model, with piecewise constant trajectories, which keep track of the sampled values. 
These are $x_d:\R_{\ge 0} \to \R^{n}$ and $z_d:\R_{\ge 0} \to \R^{n}$, and are to be seen as the time-sampled versions of $x$ and $z$ respectively. We set $y_d= h_\sigma(x_d)$ to denote the sampled values of the output that are sent to the controller, and $u_d = k_\sigma(z_d)$ to denote the sampled control values, which are sent to the plant. Whenever a new sampled value of the output (resp.~control input) is to be sent, we update $x_d^+ = x$ so that $y_d^+ = h_\sigma(x)$ (resp.~$z_d^+=z$ so that $u_d=k_\sigma(z)$). We emphasize that the plant and controller use the same value of the switching signal $\sigma$ at all times.

The overall schematic of the closed-loop system is given in Figure~\ref{fig:sampledSys}. We denote the measurement error due to sampling in the output by $d_y:= y_d - y$, and $d_z:= z_d - z$ denotes the error which appears in the plant dynamics due to sampling of the input. By constructing a Lyapunov function for the augmented system using the cascade principle, we next show that the state of system \eqref{cascadeFam3}-\eqref{samplingcontroller} converges to the origin for appropriately designed sampling algorithms.

The assumptions imposed on the nominal system \eqref{cascadeFam3}-\eqref{samplingcontroller} are now listed below:
\begin{assum}[leftmargin=*]
\item\label{alphah}
For each $p \in \cP$, there exists a class $\cK$ function $\alpha_{h,p}$ such that the function $h_p:\R^n \to \R^{n_y}$ satisfies: 
\[
|y| = |h_p(x)| \leq \alpha_{h,p}(|x|), \quad \forall p \in \cP, \forall x \in \R^n.
\]
\end{assum}

\begin{lyap}[leftmargin=*]
\item\label{spvop} There exist continuously differentiable functions $V_{o,p} : \R^n \to \Rposi$, $p \in \cP$, class $\cK$ function $\alpha_{o,p}$, and class $\cK_\infty$ functions $\overline{\alpha}_{o,p}$, $\underline{\alpha}_{o,p}$ such that, for every $(x,z,u,y,d_y) \in \R^{2n+n_u+2n_y}$,
\begin{subequations}
\begin{align}
&\underline{\alpha}_{o,p}(|e|) \leq V_{o,p}(e) \leq \overline{\alpha}_{o,p}(|e|) \label{spvopbound} \\
&\left\langle\frac{\partial V_{o,p}}{\partial e}(e) , f_{c,p} (x, u) - f_{o,p}(z, u, y + d_y) \right\rangle \leq \notag \\
&\quad - \alpha_{o,p}(V_{o,p}(e)) + \gamma_{o,p}(|d_y|), \label{spdotvopbound}
\end{align}
\end{subequations}
where $e = z - x$.
\end{lyap}
%Note that $|d_y|$ is the dynamic error of output caused by the fact of sampling $d_y = y - y_d$ that will be discussed later. 

\begin{lyap}[leftmargin=*]
\item\label{spvcp} There exist continuously differentiable functions $V_{c,p} : \R^n \to \Rposi$, $p \in \cP$, class $\cK$ functions $\alpha_{c,p}$, $\gamma_{c,p}$, and class $\cK_\infty$ functions $\overline{\alpha}_{c,p}$, $\underline{\alpha}_{c,p}$ such that
\begin{subequations}
\begin{align}
&\underline{\alpha}_{c,p}(|x|) \leq V_{c,p}(x) \leq \overline{\alpha}_{c,p}(|x|)\label{spvcpbound} \\
&\left\langle \frac{\partial V_{c,p}}{\partial x} , f_{c,p}(x, k_p(x+e+d_z)) \right\rangle \leq - \alpha_{c,p}(V_{c,p}(x)) \notag \\ & \quad \quad + \gamma_{c,p}(V_{o,p}(e)) + \gamma_{c,p}(|d_z|),\label{spdotvcpbound} 
\end{align}
\end{subequations}
hold for every $(x,z,u,y,d_y) \in \R^{2n+n_u+2n_y}$.
\end{lyap}
%Note that $|d_z|$ is the dynamic error of estimated state caused by the fact of sampling $d_z = z - z_d$ that will be discussed later. 

For the class of plants and controllers satisfying the aforementioned hypotheses, we are now interested in designing the sampling algorithms, and characterizing the class of switching signals which result in an overall asymptotically stable system.

\subsection{Sampling Algorithms}

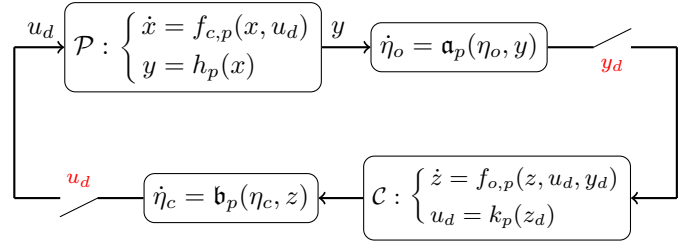
\begin{figure}[!t]
\begin{center}
\begin{tikzpicture}[circuit ee IEC,
every info/.style={font=\footnotesize}]
%system
\draw (-2,0) node [rectangle, rounded corners, draw, minimum height =0.65cm, text centered] (sys) {$ \cP: \left\{ \begin{aligned}\dot x & = f_{c,p}(x,u_d) \\ y&= h_p(x)\end{aligned}\right .$};
\draw (sys.east) node[anchor=south west] {$y$};
\draw (sys.west) node[anchor=south east] {$u_d$};
\draw (2,-2) node [rectangle, rounded corners, draw, minimum height = 0.65cm, text centered] (obs) {\small $\cC: \left\{ \begin{aligned} &\dot z = f_{o,p}(z,u_d,y_d)\\  &u_d= k_p(z_d)\end{aligned}\right .$};
\draw (1.5,0) node [rectangle, rounded corners, draw, minimum height = 0.65cm, text centered] (zO) {$\dot \eta_o = \mathfrak{a}_p(\eta_o,y)$};
\draw (-1.5,-2) node [rectangle, rounded corners, draw, minimum height = 0.65cm, text centered] (zC) {$\dot \eta_c = \mathfrak{b}_p(\eta_c,z)$};
%arrows
\coordinate (tr) at ([xshift=1.7cm)]zO.east);
\coordinate (br) at ([yshift=-2cm)]tr);
\coordinate (bl) at ([xshift=-1.7cm)]zC.west);
\coordinate (tl) at ([yshift=2cm)]bl);
\draw [thick,->] (sys.east)--(zO.west);
\draw [thick,->] (obs.west)--(zC.east);
\draw [thick] (zO.east) to [make contact={info'={[red]$y_d$}}] (tr);
\draw [thick] (tr)--(br);
\draw [thick,->] (br) -- (obs.east);
\draw [thick] (zC.west) to [make contact={info'={[red]$u_d$}}] (bl);
\draw [thick] (bl)--(tl);
\draw [thick,->] (tl) -- (sys.west);

\end{tikzpicture}
\caption{Nonlinear switched systems with sampled-data control.}
\label{fig:sampledSys}
\end{center}
\end{figure}

As mentioned in the introduction, we are interested in analyzing the stability of the closed-loop system under event-based sampling rules.  To do so, the auxiliary signals $x_d$, $z_d$ are thus modeled as
\begin{align*}
\begin{cases}
\dot x_d = 0,\\
\dot z_d = 0,
\end{cases}
\qquad
\begin{cases}
x_d^+ = x, & \text{if $\et_1 = \true$}\\
z_d^+ = z, & \text{if $\et_2 = \true$}
\end{cases}
\end{align*}
and by setting $y_d = h_\sigma(x_d)$ and $u_d = k_\sigma(z_d)$, the dynamics of the system with time-sampled inputs and outputs are given by
\begin{align*}
\dot x & = f_{c,\sigma} (x, u_d) =  f_{c,\sigma} (x, k_\sigma(z_d))\\
\dot z &= f_{o,\sigma}(z, u_d, y_d) = f_{o,\sigma}(z, k_\sigma(z_d), h_\sigma(x_d)).
\end{align*}
%We want to study conditions on the switching signals, and describe the events for updating $y_d$ and $u_d$ which result in asymptotic stability of the origin for $(x,z)$ system.

To define the events at which the outputs and inputs are updated, we introduce the following dynamic filters:
\begin{subequations}\label{dotetaoetac}
\begin{align}
\dot \eta_o &:= -\beta_{o,p}(\eta_o) + \rho_{o,p}(|y|) + \gamma_{o,p}(|h_p(x) - h_p(x_d)|)\label{defdotetao} \\
\dot \eta_c &:= -\beta_{c,p}(\eta_c) + \rho_{c,p} \left(\frac{|z|}{2}\right) + \gamma_{c,p}(|z- z_d|)
\end{align}
\end{subequations}
where, for each $p \in \cP$, $\beta_{o,p},\beta_{c,p}, \rho_{o,p}, \rho_{o,p}, \gamma_{o,p}, \gamma_{c,p}$ are class $\cKinfty$ functions, and the initial conditions for $\eta_o$ and $\eta_c$ are chosen to be some positive numbers. We say that
\begin{equation}\label{eq:jumpCond}
\begin{cases}
\et_1 = \true & \text{if } \vert y - y_d\vert \ge \mu_{o,\sigma}(\eta_o)\\
\et_2 = \true & \text{if } \vert z - z_d \vert \ge \mu_{c,\sigma}(\eta_o),
\end{cases}
\end{equation}
for some class $\cK_\infty$ functions $\mu_{o,p}, \mu_{c,p}$.
Note that $\et_1$ and $\et_2$ may occur at different times, or simultaneously, and that each one corresponds to a different update rule.
%The overall schematic of the closed-loop system with such sampling algorithms is given in Figure~\ref{fig:sampledSys}.

\subsection{Hybrid Model}
Just as we did in Section~\ref{sec:plainHybSys}, it is convenient to write the entire system with controlled plant dynamics, controller, and sampling algorithms using the framework of hybrid systems.
To do so, we first introduce
 \[
 \xi := (x, z, x_d, z_d, \eta_o, \eta_c, p, \tau) \in \R^{4n+ 2} \times \cP\times [0,N_0]
 \]
 to describe the state the closed-loop system. The variables $u_d$ and $y_d$ are obtained by setting $u_d = k_p(z_d)$ and $y_d=h_p(x_d)$.

The flow set is now described as
\begin{align}
\cC &=\left\{\xi: |y-h_p(x_d)| \leq \sigetao \right\} \cap \left\{\xi: |z-z_d| \leq \sigetac \right\} \notag \\
&\quad \cap \left\{\xi: \eta_o \geq 0 \wedge \eta_c \geq 0\right\} \cap \left\{\xi: \tau \in [0,N_0] \right\}, \notag
\end{align}
so that the state variables evolve according to a differential equation/inclusion when $\xi \in \cC$.
By construction, jumps in at least one of the state variables occur either due to switching, or when the condition for $\et_i$, $i = 1,2$, holds true.
The jump set $\cD$ therefore corresponds to the switching event, or the sampling event, and they may not occur at the same time. The jump set is defined as
\begin{equation}
\cD = \cD_{\sw} \cup \cD_{u} \cup \cD_y
\end{equation}
where
\begin{equation}
\begin{aligned}
\cD_{\sw} & := \left\{\xi: \tau \in [1,N_0]\right\} \\
\cD_{u} &:= \left\{\xi: |z-z_d| \geq \sigetac \right\} \\
\cD_y &:= \left\{\xi: |y-h_p(x_d)| \geq \sigetao \right\}
\end{aligned}
\end{equation}
so that the variables may get updated instantaneously when $\xi \in \cD$.
The evolution equations for the augmented variable $\xi$ can now be described as follows:
\begin{subequations}\label{eq:bigHybSys}
\begin{align}
&\xi \in \cC:
\begin{cases}
\dot x = f_{c,p}(x, k_p(z_d))\\
\dot z = f_{o,p}(z, k_p(z_d), h_p(x_d))\\
\dot x_d = 0 \\
\dot z_d =0 \\
\dot \eta_o = -\beta_{o,p}(\eta_o) + \rho_{o,p}(|y)|) + \gamma_{o,p}(\vert y-y_d\vert) \\
\dot \eta_c = -\beta_{c,p}(\eta_c) + \rho_{c,p}(\frac{|z|}{2}) + \gamma_{c,p}(|z- z_d|) \\
\dot p = 0 \\
\dot \tau \in [0, \frac{1}{\tau_a}]\\
\end{cases}
\end{align}
where, in the description of $\eta_o$-dynamics, we recall that $y = h_p(x)$ and $y_d(x) = h_p(x_d)$. The jump maps for this system are:
\begin{align}
&\xi \in \cD:
\begin{cases}
\xi^+ \in \cG(\xi) = \cG_{\sw} (\xi) \cup \cG_u(\xi) \cup \cG_y(\xi)
\end{cases}
%\begin{cases}
%x^+ = x \\
%z^+ = z \\
%\eta_o^+ = \eta_o \\
%\eta_c^+ = \eta_c \\
%p^+ = q \in \cP \\
%u_d^+ = k_q(z) \\
%y_d^+ = h_q(x) \\
%\tau^+ = \tau -1 .
%\end{cases} 
\end{align}
\end{subequations}
\begin{equation}
\cG_{\sw}(\xi) = \begin{bmatrix} x\\z\\ x \\ z \\ \eta_o\\ \eta_c\\ \cP\setminus\{p\} \\ \tau-1\end{bmatrix},
\cG_u(\xi) = \begin{bmatrix} x \\z \\ x_d \\ z\\ \eta_o\\ \eta_c\\ p \\ \tau\end{bmatrix},
\cG_y(\xi) = \begin{bmatrix} x\\z \\ x \\ z_d \\ \eta_o\\ \eta_c\\ p \\ \tau\end{bmatrix}.
\end{equation}
It is noted that this system satisfies the basic assumptions required for the existence of solutions \cite[Assumption~6.5]{GoebSanf09}.
We are interested in asymptotic stability of the compact target set $\cA$ defined as
\begin{equation}\label{eq:defSetA}
\cA:= \left \{ 0  \right\}^{4n+2} \times \cP \times [0,N_0]
\end{equation}
for the hybrid system \eqref{eq:bigHybSys}. Our design problem can thus be formulated as follows:

{\it Problem statement:} For each $p \in \cP$, find the design functions $\beta_{o,p},\beta_{c,p}, \rho_{o,p}, \rho_{c,p}, \gamma_{o,p}, \gamma_{c,p}, \mu_{o,p}, \mu_{c,p}$ appearing in \eqref{dotetaoetac}, \eqref{eq:jumpCond}, and the lower bound on the average dwell-time $\tau_a$, such that the set $\cA$ defined in \eqref{eq:defSetA} is globally asymptotically stable for the hybrid system \eqref{eq:bigHybSys}.

\subsection{Stability Analysis}
To state the main result of this section on asymptotic stability of the set $\cA$ for system~\eqref{eq:bigHybSys}, we introduce the design criteria that must be satisfied by the functions introduced in the sampling algorithms \eqref{dotetaoetac}, \eqref{eq:jumpCond}.
Recalling the function $\nu_p$ introduced in \eqref{defVp}, and the the hypotheses \ref{spvop}, \ref{spvcp}, the following conditions are imposed on the functions
$\beta_{o,p}$, $\beta_{c,p}$, $\mu_{o,p}$, $\mu_{c,p}$, $\rho_{o,p}$ and $\rho_{c,p}$, for each $p \in \cP$:
\begin{design}[leftmargin=*]
\item\label{Dalphah} $\beta_{o,p}$ and $\beta_{c,p}$ are differentiable functions of class $\cK$.
\end{design}
\begin{design}[leftmargin=*]
\item\label{Dsigmaoc} Let $\theta_p$ be a function of class $\cKinfty$ defined as:
\[
\theta_{p}(s) := \alpha_{o,p}\inv (2\gamma_{o,p}(s)).
\]
Choose the functions $\mu_{o,p}$ and $\mu_{c,p}$ such that, for some $\lambda \in (0,1)$,
\[
(\gamma_{o,p} \circ \mu_{o,p})(s) [1 + (\nu_p \circ \theta_{p} \circ \mu_{o,p})(s)] \leq (1- \lambda) \beta_{o,p}(s)
\]
\[
2(\gamma_{c,p} \circ \mu_{c,p}) (s) \leq (1 - \lambda)\beta_{c,p}(s).
\] 
\end{design}

\begin{design}[leftmargin=*]
\item\label{Drhoc} The functions $\rho_{o,p}$ and $\rho_{c,p}$ in \eqref{dotetaoetac} are positive definite and are chosen such that for each $s \geq 0$: 
\[
\rho_{o,p} \circ \alpha_{h,p} \circ \underline{\alpha}_{c,p}\inv(s) \leq 0.5(1 - \lambda)\alpha_{c,p}(s)
\] \[
\rho_{c,p}(s) \leq (1-\lambda) \min \left\{ \gamma_{c,p}(s), 0.5 \, \alpha_{c,p}(\underline{\alpha}_{c,p}(s)) \right\}.
\]
\end{design}

It can be guaranteed that there always exists a solution to the inequalities in \ref{Dalphah}, \ref{Dsigmaoc}, \ref{Drhoc} using the properties of $\cK_\infty$ functions and the results given in \citep[Corollary 3.2]{GeisWirt14} and \citep{Kell14}.

To state the main result, we recall the definition of $\ell_p$ from \eqref{defl} and choose $\widetilde{\alpha}_p \in \cK_\infty$ such that $\widetilde{\alpha}_p(s) = $
\[
\min \left\{ \beta_{o,p}\left(\qtr s \right), \beta_{c,p}\left(\qtr s \right), \alpha_{c,p}\left(\qtr s \right), \gamma_{c,p}  \left( \qtr \ell_p \inv (s) \right) \right\}. 
\]
The function $\psi$ is chosen to be a differentiable $\cK_\infty$ function, with $\psi'(0) = 0$, and
\begin{subequations}\label{eq:psiSamp}
\begin{align}
& \psi(s) \le \min\{c_0s, \widetilde \alpha_p(s) \, \vert \, p \in \cP\}, \quad {\color{blue} s \ge 0},\\
{\color{blue}\text{\sout{and}}} \quad & {\color{blue}\hbox{\sout{$\psi(s) \le \min\{\widetilde \alpha_p(s) \, \vert \, p \in \cP\}, \quad s \ge c_1$}}},
\end{align}
\end{subequations}
for some $c_0> 0$. Finally, let
\begin{equation}\label{eq:chiSamp}
\chi(s) :=  \max_{p,q \in \cP} \left\{ \ell_q \circ \overline{\alpha}_{o,q} \circ \ul\alpha_p\inv (s) + \overline{\alpha}_{c,q} \circ \ul\alpha_p\inv (s) \right\}.
\end{equation}
%to state the result on stability of the sampled-data closed-loop system.
%It follows from the Proposition~\ref{basicVp} and the fact that the updated states in jump set remain constant that
%for $p \in \cP$, we define the function $\chi \in \cKinfty$ that for $s \geq 0$,
%\begin{equation}\label{defchi}
%\chi(s) = \max_{p,q \in \cP} \left\{ \alpha_{2,p}  \circ \alpha_{1,q}\inv (s) \right\},
%\end{equation}
%where $\alpha_{1,p}$ defined as in~\eqref{defalpha1p}, and 
%\[
%\alpha_{2,p} |(x,e)| =  l \circ \overline{\alpha}_{o,p}(|e|) + \overline{\alpha}_{c,p}(|x|).
%\]

\begin{thm}\label{thm:stabSampSys}
Consider the system \eqref{eq:bigHybSys} and assume that \ref{alphah}, \ref{Assumq}, \ref{spvop}, \ref{spvcp} hold. Suppose that the sampling algorithms \eqref{dotetaoetac} and \eqref{eq:jumpCond} are designed such that \ref{Dalphah}, \ref{Dsigmaoc}, \ref{Drhoc} are satisfied for some $\lambda \in (0,1)$. If the average dwell-time $\tau_a$ satisfies: 
\begin{equation}\label{condADTSamp}
\lambda \tau_a > \zeta^*:= \sup_{s \ge 0} \, \int_{s }^{\chi(s)} \frac{dr}{\psi(r)}
\end{equation}
for $\psi$ and $\chi$ given in \eqref{eq:psiSamp} and \eqref{eq:chiSamp}, then the set $\cA$ given in \eqref{eq:defSetA} is globally asymptotically stable for the system \eqref{eq:bigHybSys}.
\end{thm}
The fundamental idea behind the proof is to first construct a weak Lyapunov function $W$ for system \eqref{eq:bigHybSys} with respect to set $\cA$ in \eqref{eq:defSetA}. Using additional arguments based on cascade hybrid systems, it is then shown that the solutions along which the derivative of $W$ is possibly zero, also converge to the set $\cA$.
\begin{proof}
We start with the function
\begin{equation}\label{constwp2}
W(\xi) := \exp{({\color{blue}2 c_0}\zeta \tau)} \varphi(V_p(x,e) + \eta_o + \eta_c) 
\end{equation}
where $\zeta \in (\zeta^* , \lambda \tau_a)$ with $\zeta^* $ given in \eqref{condADTSamp}, $\varphi$ is defined in \eqref{eq:defPhi}, and the function $V_p$ is defined as in \eqref{defVp}. It is noted that $W$ does not involve the variables $(x_d, z_d) \in \R^{2n}$, so we only have the bounds
\begin{equation}\label{eq:Wsamp}
\ul\alpha (\vert \xi \vert_{\cA_c}) \le W (\xi) \le \ol\alpha (\vert \xi \vert_{\cA_c})
\end{equation}
where $\cA_c := \{0\}^{2n} \times \R^{2n} \times \R_{\ge 0}^2 \times \cP \times [0,N_0]$, and $\ul\alpha$, $\ol\alpha$ are some class $\cK_\infty$ functions.
When $\xi \in \cC$, we have the derivative of $W$:
\begin{equation}\label{eq:dotWpSamp}
\dot W = W \left[2{\color{blue}c_0} \zeta \dot \tau + \frac{2{\color{blue}c_0}} {\psi(V_p + \eta_o + \eta_c) }  \left(\dot V_p + \dot \eta_o + \dot \eta_c \right) \right].
\end{equation}
To show that $\dot W$ is bounded by a negative semidefinite function, we compute bounds on $\dot V_p, \dot \eta_o$ and $\dot \eta_c$ in the flow set $\cC$.

Using \ref{Assumq}, \ref{spvop}, \ref{spvcp} and the inequality derived in \eqref{eq:dlVo}, we obtain
\begin{multline}\label{dotVpsample}
\dot V_p \leq -\demi \nu_p(V_{o,p}(e)) \alpha_{o,p}(V_{o,p}(e))\\
  + \nu_p(\theta_p(\vert y - y_d\vert))\gamma_{o,p}(|y - y_d|)  \\
 - \alpha_{c,p}(V_{c,p}(x)) + \gamma_{c,p} (V_{o,p}(e)) + \gamma_{c,p}(|z-z_d|).
\end{multline}
It follows from the definition of $\nu_p$, sampling condition \eqref{eq:jumpCond}, and \eqref{dotVpsample} that 
\begin{align}\label{spdotvp}
\dot V_p &\leq - \gamma_{c,p} (V_{o,p}(e)) + \nu_p(\theta_p(\sigetao)) \gamma_{o,p}(\sigetao)\notag \\
&\quad - \alpha_{c,p}(V_{c,p}(x)) + \gamma_{c,p}(\sigetac).
\end{align}
The derivative of $\eta_o$ is seen to satisfy
\begin{align}\label{dotetao}
\dot \eta_o &\leq -\beta_{o,p}(\eta_o) + \rho_{o,p} \circ \alpha_{h,p}(|x|) + \gamma_{o,p}(\sigetao) \notag \\
&\leq -\beta_{o,p}(\eta_o) + \rho_{o,p} \circ \alpha_{h,p} \circ \underline{\alpha}_{c,p}\inv (V_{c,p}(x)) + \gamma_{o,p}(\sigetao).
\end{align}
The derivative of $\eta_c$ can be bounded as follows:
\begin{align}\label{dotetac}
\dot \eta_c &\leq - \beta_{c,p}(\eta_c) + \rho_{c,p}\left( \frac{|x| + |e|}{2} \right) + \gamma_{c,p}(\sigetac) \notag \\
&\leq - \beta_{c,p}(\eta_c) + \rho_{c,p} \circ \underline{\alpha}_{c,p}\inv (V_{c,p}(x)) \notag \\
&\quad + \rho_{c,p} \circ \underline{\alpha}_{o,p}\inv (V_{o,p}(e)) + \gamma_{c,p}(\sigetac).
\end{align}
Now combining \eqref{spdotvp}, \eqref{dotetao}, \eqref{dotetac}, and using the inequalities given in \ref{Dalphah}, \ref{Dsigmaoc}, and \ref{Drhoc}, we get
\begin{align*}
\dot V_p + \dot \eta_o + \dot \eta_c &\leq -\lambda [ \beta_{o,p}(\eta_o) + \beta_{c,p}(\eta_c) + \alpha_{c,p}(V_{c,p}(x)) \notag \\
& \quad + \gamma_{c,p} \inv \circ \ell_p\inv (\ell_p(V_{o,p}(e))) ] \\
&\le -\lambda \widetilde{\alpha}_p(V_p + \eta_o + \eta_c).
\end{align*}
Substituting this expression in \eqref{eq:dotWpSamp}, and using the definition of $\psi$ in \eqref{eq:psiSamp}, we obtain
\begin{align*}
\dot W &\leq {\color{blue}c_0} W \left[ 2 \zeta \dot \tau - \frac{2 \lambda \widetilde{\alpha}_p(V_p + \eta_o + \eta_c)}{\psi(V_p + \eta_o + \eta_c)} \right] \\
&\le {\color{blue}c_0} W( 2 \zeta \dot \tau - 2 \lambda) \\
&\le {\color{blue}c_0} W \left(\frac{2 \zeta}{\tau_a} - 2 \lambda \right) \\
&= {\color{blue}c_0} W \left(\frac{2(\zeta - \lambda\tau_a )}{\tau_a} \right),
\end{align*}
which is the desired inequality for $\dot W$ over the flow set since we chose $\zeta < \lambda\tau_a$.

When $\xi \in \cD$, we calculate the maximum of $W(\xi^+)$, over the set $\cG(\xi) \ni g = \xi^+$, as follows:
\begin{align*}
\max_{g\in\cG(\xi)} W(g) &= \max_{g\in\cG(\xi)} \exp(2 {\color{blue}c_0} \zeta \tau^+) \varphi\left( V_{p^+}(x^+,e^+) + \eta_o^+ + \eta_c^+ \right) \\
&= \max_{p^+\in \cP} \exp(2 {\color{blue}c_0} \zeta \tau - 2 {\color{blue}c_0} \zeta) \varphi\left(V_{p^+}(x,e) + \eta_o + \eta_c \right). \notag \\
\end{align*}
To get a bound on the right-hand side in terms of the value of $\xi$ just prior to the jump, we recall that the function $\chi$ introduced in \eqref{eq:chiSamp} satisfies
\[
V_q \leq \chi(V_p), \quad \forall p,q \in \cP.
\]
Moreover, from the definition of $\varphi$, it follows that
\[
 \varphi\left(V_{p^+}  + \eta_o + \eta_c \right) \leq \exp \left( \int_{1}^{\chi(V_p)+\eta_o + \eta_c} \frac{2 {\color{blue}c_0} \, dr}{\psi(r)} \right).
\]
We then observe that
\begin{align}
\int_{1}^{\chi(V_p)+\eta_o + \eta_c} \frac{2 \,dr}{\psi(r)} &= \int_{V_p+\eta_o + \eta_c}^{\chi(V_p)+\eta_o + \eta_c} \frac{2 \, dr}{\psi(r)} + \int_{1}^{V_p+\eta_o + \eta_c} \frac{2 \,dr}{\psi(r)}. \notag \\
\end{align}
Since $\eta_o + \eta_c \geq 0$ and $\frac{1}{\psi(r)}$ is decreasing, we have: 
\[
 \int_{V_p + \eta_o + \eta_c}^{\chi(V_p) + \eta_o + \eta_c} \frac{2 {\color{blue}c_0} \, dr}{\psi(r)} \leq \int_{V_p}^{\chi(V_p)} \frac{2 {\color{blue}c_0} \, dr}{\psi(r)} \le 2 {\color{blue}c_0} \zeta^*,
\] 
so that
\[
\int_{1}^{\chi(V_p)+\eta_o + \eta_c} \frac{2 {\color{blue}c_0} \, dr}{\psi(r)} \leq 2 {\color{blue}c_0} \zeta^* + \int_{1}^{V_p+\eta_o + \eta_c} \frac{2 {\color{blue}c_0} \,dr}{\psi(r)}.
\]
Thus, the value of $W$ after each jump is bounded as
\begin{align}
\max_{g \in \cG(\xi)} W(g) &\leq \exp(-2{\color{blue}c_0} (\zeta-\zeta^*))\exp(2 {\color{blue}c_0}\zeta \tau) \varphi(V_p + \eta_o + \eta_c) \notag \\
&\leq \exp(-2{\color{blue}c_0} (\zeta-\zeta^*)) W(\xi) \qquad \forall\,\xi \in \cD.
\end{align}

Because of the bounds in \eqref{eq:Wsamp}, it thus follows that $\xi$ converges asymptotically to the set $\cA_c$. To conclude further that $(x_d,z_d)$ also converge to $\{0\}^{2n}$, 
one can invoke the arguments based on the LaSalle's invariance principle \cite[Corollary 8.9(ii)]{GoebSanf12}, and cascaded hybrid systems \cite[Corollary~19]{GoebSanf09}.
Following the same recipe as in \cite[Proof of Theorem~1]{TanwTeel15}, we next show that the set $\cA$ of the closed-loop system is a globally asymptotically stable (GAS) for system \eqref{eq:bigHybSys}:

{\it Step~1 -- Pre-GAS of $\{0\}$ for truncated systems:}
For a fixed initial condition, there exist compact set $\cM_1 \subset \R^{2n}$, $\cM_2 \subset \R^2 \times \cP \times \R$ such that $(x,z) \in \cM_1$ and $(\eta_o,\eta_c,p,\tau) \in \cM_2$.
Recalling that $z_d$ and $x_d$ remain constant during flows, and are reset to $z$ and $x$, which belong to a compact set, there exists a compact set $\cM_d$ such that $(x_d, z_d) \in \cM_d$.
Consider the truncation of system~\eqref{eq:bigHybSys} to the set $\cM:=\R^{2n} \times \cM_d \times \R_{\ge 0}^2 \times \cP \times [0,N_0]$, which has the flow set $\cC_{\cM}:=\cC \cap \cM$, the jump set $\cD_{\cM}:= \cD_{\cM} \cap \cM$.
For this truncated system, it follows from the invariance principle \cite[Corollary~8.9\,(ii)]{GoebSanf12} that the set $\cA_1:= \{0\}^{2n} \times \cM_d \times \{0\}^2 \times \cP \times [0,N_0]$
%$\cA_1:= \{\xi \, : \,  (x,z,\eta_o, \eta_c) \in \{0\} \, \wedge \, (x_d,z_d) \in \cM_d \, \wedge \, p \in \cP \, \wedge \, \tau \in [0,N_0]\}$
is pre-GAS.
We next invoke the stability result for cascaded hybrid systems \cite[Corollary~19]{GoebSanf09} to claim that the set $\cA$ in \eqref{eq:defSetA} is pre-GAS for the truncated system.
Indeed, for every system trajectory contained in $\cA_1$, we have $\eta_o =\eta_c = 0$, and from the definition of the sets $\cC$ and $\cD$, we must then have $x_d = 0$ and $z_d = 0$.

{\it Step~2 -- Bounded solutions and Pre-GAS of $\{0\}$ for \eqref{eq:bigHybSys}:}
As shown in the first step, for each initial condition, there exist compact sets $\cM_1$, $\cM_2$ and $\cM_d$ such that $\xi$ is contained in the compact set $\cM_1 \times \cM_d \times \cM_2$ for all times.
Boundedness of the solutions now allows us to conclude that $\cA$ is pre-GAS for the original system~\eqref{eq:bigHybSys}.
To see this, assume that there exists a solution for which $(x,z,x_d,z_d,\eta_o,\eta_c)$ does not converge to $\{0\}$. Since all solutions are  bounded, there exists a compact set $\cM_d$ such that this bounded solution eventually coincides with the solution of the system truncated to $\R^{2n} \times \cM_d \times \R_{\ge 0}^2 \times \cP \times [0,N_0]$. But, every solution of the truncated system must converge to $\cA$. Hence, for \eqref{eq:bigHybSys}, a bounded solution not converging to $\cA$ cannot exist, proving that $\cA$ is pre-GAS.

{\it Step~3 -- $\{0\}$ is GAS for \eqref{eq:bigHybSys}:} To move from pre-asymptotic stability to asymptotic stability of the compact set $\cA$, we next show that every solution of \eqref{eq:bigHybSys} is forward complete.
This is seen due to the fact that for each $\xi \in \cC \setminus \cD$, the solutions would always continue to flow. Moreover, after each jump the states are reset to the set $\cC \cup \cD$, making it possible to extend the time domain for the solutions either by jump or flow.
Hence, each solution of the system is forward complete, proving that the set $\cA$ is GAS.
%
%According to the expression of $W_p$, $W_p$ is bounded by two $\cKinfty$ functions. Then, it follows from the \cite[Theorem 3.18]{GoebSanf12} that the set $\cA$ is pre-asymptotically stable. The fact that a set of $\left\{ 0 \right\}^2$ is excluded is because that the Lyapunov function $W_p$ can only shows that the state $x, z,\eta_o,\eta_c \to \left\{ 0 \right\}^{2n+2}$ and $\tau, p \to [0,N_0] \times \cP$, but cannot tell us anything about the sampled state estimator $z_d$ and the sampled output $y_d$. At this point, one can invoke the arguments given in \cite{TanwTeel15} to establish that the set $\cA$ is in fact asymptotically stable.
\end{proof}

\begin{remark}
For implementation purposes, it is important to show that, for event-based sampling, there is a uniform lower bound on the minimal inter-sampling time between two consecutive sampling instants. For the algorithms employed in this article, such a lower bound has been obtained for nonswitched dynamical systems in \citep[Theorem~2]{TanwTeel15} under certain additional assumptions on the functions appearing in the dynamic filters \eqref{dotetaoetac}. For switched systems, when working under the slow switching assumption like dwell-time or average dwell-time, it suffices to have such a lower bound for an individual dynamical subsystem since this guarantees there will be no accumulations of jump events if the switching signal has no accumulation of switches.
\end{remark}

%%%%%%%%%%%%%%%%%% 		SIMULATION       %%%%%%%%%%%
\section{Example and Simulation Result}\label{sec:example}
As an illustration of Theorem~\ref{thm:stabSampSys}, we consider an academic example of a switched system with two modes.
The first subsystem is described by linear dynamics as follows:
\begin{equation}\label{linearswisysexample}
p = 1:  \begin{cases} \dot x= A_1 x + B_1 u \\
y = C_1 x \end{cases}
\end{equation}
The feedback controller related to this subsystem is: 
\begin{equation}
p = 1:  \begin{cases} \dot z = A_1 z + B_1 u_d + L_1(y - C_1z) \\
u = - K_1 z, \end{cases}
\end{equation}
where we choose $A_1 =\begin{bmatrix} 0.5 & -1\\ 0 & 0.5\end{bmatrix}$, $B_1 = \begin{bmatrix} 0 \\ 1\end{bmatrix} $, $C_1 = \begin{bmatrix}1 & 0\end{bmatrix}$, $L_1 = \begin{bmatrix} 3.5 \\ -3\end{bmatrix}$ and $K_1 =\begin{bmatrix}-1.5 & 2.5\end{bmatrix} $. 

The second subsystem has nonlinear dynamics described by:
\begin{equation}
p = 2:  \begin{cases}
\dot x_{1} = x_{2} + 0.25 | x_{1} |  \\
\dot x_{2} = \mbox{sat} (x_{1} )+u_d  \\
y = x_{1}. \end{cases}
\end{equation}
The notation $\mbox{sat}$ denotes the saturation function $\mbox{sat} (x_1) = \min \left\{ 1, \max \left\{ -1, x_{1}\right\} \right\}$. The corresponding feedback controller is:
\begin{equation}
p = 2:  \begin{cases}
\dot z_{1} = z_{2} + 0.25 |y| + l_{1}\left(y - z_{1} \right) \\
\dot z_{2} = \mbox{sat}(y) + u_d + l_{2} \left( y - z_{1} \right) \\
u = \mbox{sat}(z_{1}) + K_2 z, \end{cases}
\end{equation}
where we choose $K_2  = \begin{bmatrix} -2 & -2\end{bmatrix} $ and $L_2 =\begin{bmatrix}  l_{1} \\  l_{2}\end{bmatrix}= \begin{bmatrix} 2 \\ 2\end{bmatrix}$.

%\begin{figure}
%\centering
%\includegraphics[width=0.48\textwidth]{samplingErrors.pdf}
%\vskip -0.5em
%\caption{Simulation results: In the top plot, whenever $|y(t)-y_d(t)|$ reaches the sampling threshold $\overline \sigma_{o,p}\sqrt{\eta_o(t)}$, $y_d$ is updated. The middle plot shows $\overline \sigma_{c,p}\sqrt{\eta_c(t)}$ and $|z(t)-z_d(t)|$. The bottom plot shows the switching signal used in the simulations.}
%\label{fig:simExSamp}
%%\vskip -1.75em
%\end{figure}
%
%\begin{figure}
%\centering
%\includegraphics[width=0.48\textwidth]{lyapFcns.pdf}
%\vskip -0.5em
%\caption{The time evaluation of Lyapunov functions $V_o(e)$ and $V_c(x)$ which are not always decaying.}
%\label{fig:simExLyap}
%%\vskip -1.75em
%\end{figure}

For both subsystems, we introduce the same form of Lyapunov function: $V_{c,p}(x) = x^T P_{c,p} x$ and $V_{o,p}(e) = e^T P_{o,p} e$. Since the controller is driven by the sampled output $y_{d}$, we have for each $p \in \left\{ 1, 2\right\}$:
\[
\dot V_{o,p} \leq -a_{o,p} V_{o,p}(e) + \gmaopupbar \left| y - y_d \right|^2,
\]
where
\[
a_{o,p} = \frac{\lambda_{\min}(Q_{o,p})}{2\lambda_{\max}(P_{o,p})} \quad \text{ and } \quad \gmaopupbar = \frac{2\| P_{o,p}L_p \|^2}{\lambda_{\min}(Q_{o,p})}.
\]
Similarly, for each $p \in \left\{ 1, 2\right\}$:
\[
\dot V_{c,p} \leq -a_{c,p} V_{c,p}(x) + \gmacpupbar V_{o,p}(e)+ \gmacpupbar \left| z- z_d \right|^2,
\] 
where 
\[
a_{c,p} = \frac{\lambda_{\min}(Q_{c,p})}{2\lambda_{\max}(P_{c,p})},
\]
and
\[
\gmacpupbar =  \frac{4 \| P_{c,p} B_p K_p \|^2}{\lambda_{\min}(Q_{c,p})} \max\left\{ 1,\frac{1}{\lambda_{\min}(P_{o,p})}\right\}.
\]
\begin{figure}[!t]
\centering
\includegraphics[width=0.48\textwidth]{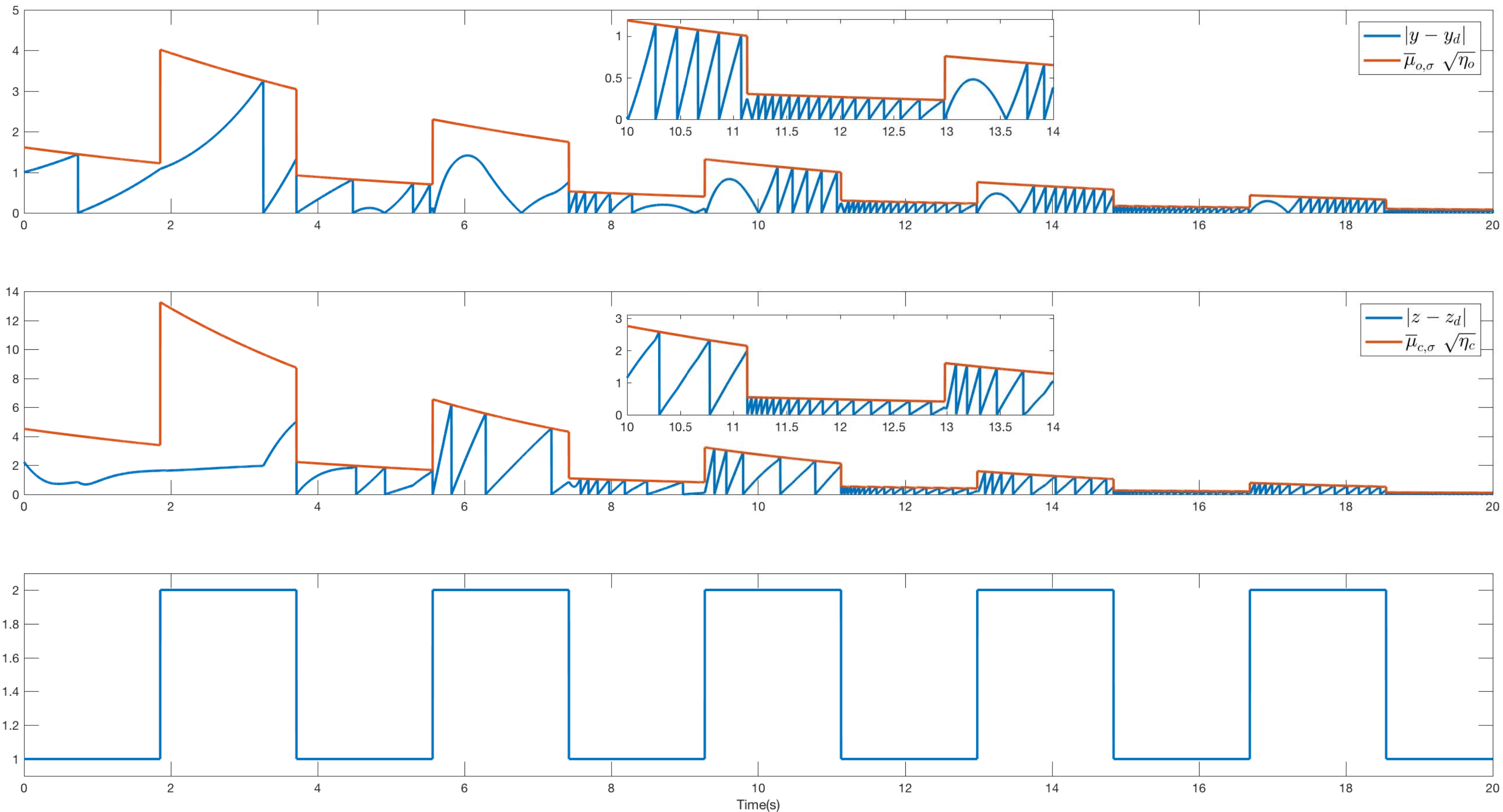}
\vskip -0.5em
\caption{Simulation results: In the top plot, whenever $|y(t)-y_d(t)|$ reaches the sampling threshold $\overline \mu_{o,\sigma}\sqrt{\eta_o(t)}$, $y_d$ is updated. The middle plot shows $\overline \mu_{c,\sigma}\sqrt{\eta_c(t)}$ and $|z(t)-z_d(t)|$. The bottom plot shows the switching signal used in the simulations.}
\label{fig:simExSamp}
%\vskip -1.75em
\end{figure}
For the dynamic filters, we choose
\begin{align}
&\dot \eta_o = -a_{o,p}\eta_o + \rhoopupbar|y|^2 + \gmaopupbar|y-y_d|^2 \notag \\
& \dot \eta_c = -a_{c,p}\eta_c + \rhocpupbar \frac{|z|^2}{4} + \gmacpupbar|z-z_d|^2, \notag
\end{align} 
where we let 
\begin{align*}
&\rhoopupbar := \frac{(1-2\eps)\lambda_{\min} (P_{o,p})}{\norm{C_p}^2},\\
&\rhocpupbar :=  \min \left\{ (1-\eps)\gmacpupbar, \eps a_{c,p}\lambda_{\min} (P_{c,p}) \right\}
\end{align*}
for some small $\eps \in (0,0.5)$.
The jump set which describes the conditions when the sampled values get updated or when there is a switching event occurs, is defined as follows:
\begin{multline*}
\cD = \left\{\xi : |y-y_d| \geq \sigmaopupbar \sqrt{\eta_o} \right\} \cup \left\{ \xi : |z-z_d| \geq \sigmacpupbar \sqrt{\eta_c}\right\} \\ \cup \left\{ \tau \in [1,N_0]\right\} 
\end{multline*}
where $\sigmaopupbar := \frac{(1-\eps) \alpha_{o,p}}{(1+\overline{\nu}_p)\gmaopupbar}$ and $\sigmacpupbar := \frac{(1-\eps)\alpha_{c,p}}{2\gmacpupbar}$, with
\[
\overline{\nu}_p = \frac{4\overline{\gamma}_{c,p}}{\lambda_{\min} (P_{o,p})}.
\]
The function $\chi$ can thus be defined as $\chi (s)= \ol \chi s$, where
\[
\overline{\chi} = \max_{p,q \in \left\{ 1, 2 \right\}} \left\{ \frac{\overline{\nu}_p \lambda_{\max} (P_{o,p}) +\lambda_{\max} (P_{c,p})}{\overline{\nu}_q \lambda_{\min} (P_{o,q}) +\lambda_{\min} (P_{c,q}) } \right\}.
\]
It follows from Theorem~\ref{thm:stabSampSys} that, if the average dwell-time $\tau_a$ satisfies that:
\[
\tau_a > \frac{\ln \overline{\chi}}{\eps},
\]
then we have the asymptotic stability of the origin for $(x,z)$ system.
The simulation results reported in Figure~\ref{fig:simExSamp} indeed show the convergence of $(y,z,\eta_o,\eta_c)$ to the origin.

\section{Conclusion}\label{sec:conc}

In this article, the construction of ISS Lyapunov functions is considered for switched nonlinear systems in cascade configuration. The stability analysis for the resulting hybrid systems is carried out under an average dwell-time condition on the switching signal, and an asymptotic ratio condition for establishing ISS. The results pave the path for studying the stabilization of switched systems with dynamic output feedback. A Lyapunov function similar to the one used for non-sampled ISS system is constructed to design the sampling algorithms and for the analysis of the closed-loop hybrid systems with sampled measurements. The results are illustrated with the help of examples and simulations. One of the limitations of the dynamic output feedback problem considered in this paper is that the controller requires exact knowledge of the switching signal. It is of interest to develop theoretical tools when there is mismatch in the switching signal between the plant and the controller. One can also consider additional measurement errors, for example, due to quantization of output and input in space, as done in \citep{TanwPrie16}. One could also potentially study the affect of random uncertainties, on top of event-based samples, as has been recently proposed in \citep{TanwTeel17}.

\section*{Acknowledgements}

The authors would like to thank Dr.~Guosong Yang for his useful comments on an earlier version of this paper.

\bibliography{localEventTrig}

\end{document}